\documentclass[twocolumn,pra,superscriptaddress,tightlines]{revtex4-2}
\pdfoutput=1
\usepackage{amsmath}
\usepackage{amssymb}
\usepackage{amsfonts}
\usepackage{amsthm}
\usepackage{mathrsfs}
\usepackage{bm, dsfont}
\usepackage{booktabs}
\usepackage{multirow}
\usepackage[usenames,dvipsnames]{color}
\usepackage{colortbl}
\usepackage[table]{xcolor}
\usepackage{enumitem}
\usepackage{dcolumn}
\usepackage{epstopdf}
\usepackage{graphicx}
\usepackage{natbib}
\usepackage[colorlinks=true,linkcolor=blue,citecolor=blue,urlcolor=blue]{hyperref}
\usepackage{hypcap}
\usepackage{verbatim, float}
\usepackage{psfrag}
\usepackage[normalem]{ulem}
\usepackage{physics}		% for braket, etc.
\usepackage{siunitx, ifsym, marvosym}
\usepackage{cleveref}

\newcommand{\id}{\mathds{1}}

\newcommand{\cE}{\mathcal{E}}
\newcommand{\cF}{\mathcal{F}}
\newcommand{\cH}{\mathcal{H}}
\newcommand{\cM}{\mathcal{M}}

\renewcommand{\t}[1]{\mathrm{#1}}
\newcommand{\be}{\begin{equation}}
\newcommand{\ee}{\end{equation}}

\newcommand{\nc}{\textup{\,\o\,}}

\newtheorem{lemma}{Lemma}

\definecolor{navyblue}{rgb}{0,0.43,0.72}

\begin{document}

\title{Towards the device-independent certification of a quantum memory}

\author{Pavel Sekatski}
\affiliation{Department of Applied Physics, University of Geneva, Geneva, Switzerland}
\author{Jean-Daniel Bancal}
\affiliation{Université Paris-Saclay, CEA, CNRS, Institut de physique théorique, 91191, Gif-sur-Yvette, France}

\author{Marie Ioannou}
\affiliation{Department of Applied Physics, University of Geneva, Geneva, Switzerland}
\author{Mikael Afzelius}
\affiliation{Department of Applied Physics, University of Geneva, Geneva, Switzerland}
\author{Nicolas Brunner}
\affiliation{Department of Applied Physics, University of Geneva, Geneva, Switzerland}
\begin{abstract}
    Quantum memories represent one of the main ingredients of future quantum communication networks. Their certification is therefore a key challenge. Here we develop efficient certification methods for quantum memories. Considering a device-independent approach, where no a priori characterisation of sources or measurement devices is required, we develop a robust self-testing method for quantum memories. We then illustrate the practical relevance of our technique in a relaxed scenario by certifying a fidelity of $0.87$ in a recent solid-state ensemble quantum memory experiment. More generally, our methods apply for the characterisation of any device implementing a qubit identity quantum channel. 
\end{abstract}
\maketitle

\section{Introduction}

A quantum memory is a central element for future quantum technologies, in particular for quantum networks~\cite{Kimble2008,Simon2017,wehner2018quantum}. Such a device allows to locally store, and relay  quantum information carried by travelling photons. Quantum memories typically involve an atomic media, interfaced with light, and a broad range of experimental platforms are currently explored, see e.g.~\cite{Tittel2010, Sangouard2011}.

A key challenge towards the realisation of a practical quantum memory, and more generally all quantum technologies, is their certification. That is, given an actual device, consisting of a complex physical system, how can we ensure its correct functioning. Moreover, this certification procedure will involve other devices, e.g.~a source for producing specific quantum states of light as well as a measurement device, which may feature their own technical imperfections and must therefore also be characterized. Attempting to certify these devices on their own would require access to other certified devices, and so on. 

Remarkably, it turns out that there exist an elegant solution out of this seemingly daunting task. Quantum theory allows for certification techniques that are ``device-independent''. That is, the correct operation of a quantum device can be verified without requiring an a priori certification of any of the devices used in the protocol (including a possible source and measurement devices); rather, it is inferred directly from the statistics observed in a black-box scenario. When certifying a device amounts to characterize it fully, this concept is known as self-testing \cite{MayersYao,SupicReview}.

So far, these ideas have been almost exclusively investigated from a purely theoretical and abstract point of view; see \cite{Tan2017,Bancal_2021} for first experiments. In recent years, a considerable understanding of the possibilities and limits of self-testing protocols has been acquired, see e.g.~\cite{magniez,McKague2012,Yang,Yang14,jed,Coladangelo_2017,Bancal_BSM,Renou2018,Tavakoli2018,Mancinska23}. While self-testing protocols typically involve the certification of certain entangled states and some sets of local measurements, methods for the self-testing of quantum channels have recently been developed~\cite{blocks,Wagner2020}. Channel certification in partially device-independent scenarios have been discussed as well~\cite{Pusey2015,DallArno2017,Rosset2018,Yu2021}.

In this work, we demonstrate how self-testing techniques can be applied in practice to certify a real-world quantum channel such as a solid-state quantum memory. The scheme requires a source of entangled photon pairs. One photon is then stored in the quantum memory. Finally, local measurements are performed on the photon retrieved from the quantum memory, as well as on the other twin photon. Based on the measurement statistics, a quantitative statement about the quality of the quantum memory is made. More precisely, we provide a lower bound both on the fidelity of the implemented quantum memory with respect to an ideal one and on its efficiency. For deterministic quantum memories, our fidelity bound is tight. In principle, these bounds do not require any assumption about the source, the memory, or the measurement devices, but they are only conclusive when the rate of particle loss is low enough, a regime that is out of reach of most current experiments. Therefore, we proceed to show how an additional simple assumption on the physics of the measurement devices allows one to obtain bounds that are applicable to setups with low detection efficiency as well. Importantly, the presented methods can be used for any device or protocol aiming to implement the identity channel on a qubit.

\section{The quality of a quantum memory}

To present our protocol for certifying a quantum memory, we first need to specify what one expects form a ``good'' quantum memory. An ideal quantum memory is a device that at time $t_0$ receives a qubit in any state and outputs a qubit in the same state on demand at a later time $t_1$. Formally an ideal memory is thus described by an identity channel between two qubits (at different times).

A device that is closer to current technological capabilities is a lossy quantum memory. The output state is only identical with the input qubit state with probability $P_\checkmark$,  otherwise the output state is \textit{orthogonal} to the qubit subspace. For example, for photonic qubits (e.g.~polarized single-photons) the memory sometimes outputs the vacuum state (when the photon was lost). It is in principle possible to filter the output, verifying that the qubit was not lost without disturbing its state. This defines an announced quantum memory -- a two-branch quantum instrument which acts as an identity channel with probability $P_\checkmark$ and otherwise reports a failure. This is our target device.

Any realistic device will perturb the input state to some extent. To describe this, consider an instrument which maps a qubit state 
$\rho$ at $t_0$ to a (possibly different) qubit state $\t{QM}_\checkmark[\rho]$ at $t_1$ with some probability $\tr \t{QM}_\checkmark[\rho]$, and otherwise reports a failure. 
To quantify how similar it is to the ideal announced memory we look at two following properties. First, in the success branch it has to be similar to the identity channel. This can be  quantified by the conditional Choi-Jamio\l{}kowski fidelity
\be
\cF_\t{id}(\text{QM}_\checkmark) \equiv F \left( \frac{(\text{id}\otimes \text{QM}_\checkmark)[\Phi^+]}{\tr\, (\text{id}\otimes \text{QM}_\checkmark)[\Phi^+]},\Phi^+ \right),
\ee
where $\Phi^+=\ketbra{\Phi^+}$ is the maximally entangled two qubit state $\ket{\Phi^+}=\frac{1}{\sqrt{2}}\left(\ket{00}+\ket{11} \right)$, and $F$ is the fidelity between two quantum states that we define here as $F(\rho,\sigma)=\left(\tr |\sqrt{\rho} \sqrt{\sigma}|\right)^2$\footnote{One note that the definition here differs from the on in \cite{blocks}, where the quantity is not squared}. 
Second, it is also important to consider the average success probability of the device
\be
P_\checkmark [\t{QM}_\checkmark]\equiv \tr \t{QM}_\checkmark\left[\frac{1}{2}\mathds{1}\right]  =
\tr\, (\text{id}\otimes \text{QM}_\checkmark)[\Phi^+] \,,
\ee
which is one in the case of a deterministic memory. To certify a quantum memory we thus require to lower bound both quantities $\cF_\t{id}(\text{QM}_\checkmark)$ and $P_\checkmark[\t{QM}_\checkmark]$.

\section{Black-box certification}

Here we look for a ``black-box'' certification protocol, where the internal modeling of each box is a priori unknown. In a black-box description, a quantum memory is a device that maps between well identified time-delayed input and outputs system, but one ignores the form of the mapping and any characteristics of the input and output systems. Hence, a memory is given by a channel, i.e.~a completely positive trace preserving (CPTP) map,
\be\label{eq:realQM}\begin{split}
    \widetilde{\t{QM}}: B(\cH^{(i)}) &\to B(\cH^{(o)}) \\
    \rho & \mapsto \widetilde{\t{QM}}[\rho].
\end{split}
\ee
Here $\rho \in B(\cH^{(i)})$ is any state of the input system -- containing all the degrees of freedom affecting the physics of the memory at time $t_0$, and defined on some Hilbert space $\cH^{(i)}$. $\widetilde{\t{QM}}[\rho] \in B(\cH^{(o)})$ is the state of the system gathered from the memory at time $t_1$, with the  associated Hilbert space $\cH^{(o)}$. According to quantum physics any real-world device with well identified input and output systems abides to the model of Eq.~\eqref{eq:realQM}.

To certify a ``black-box'' lossy quantum memory $\widetilde{\t{QM}}$  we build on the operational framework of \cite{blocks}. Precisely, we identify CPTP injection and extraction maps
\be\begin{split}\label{eq: inj ext}
V : B(\mathds{C}^2) \to B(\cH^{(i)}) \qquad
\Lambda :  B(\cH^{(o)}) \to B(\mathds{C}^2),
\end{split}
\ee
allowing us to compare the composite channel $\Lambda \circ \widetilde{\t{QM}}\circ V $
with the target announced memory for qubits. These maps identify a way in which the actual tested device can be used in order to behave similarly to the target operation. Since we want to be able to certify memories which may have low efficiency and/or are probabilistic, we will in some cases allow the extraction $\Lambda$ and injection $V$ maps to be probabilistic, or completely positive trace non-increasing (CPTN) to be formal. In this case, the target memory is successful when all the maps are successful. In general, our goal is thus to find the maps $V$ and $\Lambda$ that maximize the values of
\begin{align}\label{eq: F mem}
\cF_\text{id}^{(\Lambda,V)}\left(\widetilde{\t{QM}}\right)&=\cF_{\t{id}}\left(\Lambda \circ \widetilde{\t{QM}}\circ V\right) \\
\label{eq: P mem}
P_\checkmark^{(\Lambda,V)}\left(\widetilde{\t{QM}}\right) &= \tr \left(\Lambda \circ \widetilde{\t{QM}}\circ V \right)\left[\frac{1}{2}\mathds{1}\right].
\end{align}

In addition to the quantum memory $\widetilde{\t{QM}}$ itself our certification setup involves an entanglement source and local measurements. The source is used to prepare a bipartite state to be stored in the memory. This state can be described by a density operator $\rho_{AB}^{(i)}\in B(\cH_A\otimes \cH^{(i)}_B)$. The memory stores the system $B$ only. After storage, the state output by the memory is of the form $\rho_{AB}^{(o)} = (\text{id}\otimes \widetilde{\t{QM}})[\rho_{AB}^{(i)}] \in B(\cH_A\otimes \cH^{(o)}_B)$. The local measurements are performed on both systems $A$ and $B$. In our protocol, each party $A$ and $B$ chooses between two binary measurements, labeled by the setting $x$ and $y=0,1$ respectively. In quantum physics these are modeled by 2-outcome positive operator valued measures (POVM): $\cM_{A}=\{A_{0|x},A_{1|x}\}$ for the system $A$, and $\cM_{B}^{(i/o)}=\{B_{0|y}^{(i/o)},B_{1|y}^{(i/o)}\}$ for the system $B$ after the source $(i)$ or after the memory $(o)$. 

In summary, the a priori quantum model of the setup involves the following elements  $\{  \rho_{AB}^{(i)},\rho_{AB}^{(o)}, \cM_A, \cM_B^{(i)},\cM_B^{(o)}\}$
defined on the Hilbert spaces $\cH_A, \cH_B^{(i)}$ and $\cH_B^{(o)}$. Given a description of all states and measurements we define the expected correlations 
\be\label{eq: P}
\t{P}^{(i/o)}(a,b|x,y)= \t{tr}\,   \big(A_{a|x}\otimes B_{b|y}^{(i/o)}\big) \rho_{AB}^{(i/o)}.
\ee
The protocol consists in repeatedly measuring at random the states prepared by the source ($\rho_{AB}^{(i)}$) or output by the memory ($\rho_{AB}^{(o)}$) with randomly chosen measurement settings $x$ and $y$. Assuming that all the trials are identical this allows one to estimate the correlations $\t{P}^{(i/o)}(a,b|x,y)$.

\section{Results}

We are now ready to formulate the main results of this paper, namely how the above setup can be used to certify and characterize a quantum memory. We consider three different scenarios depending on whether the memory and the considered maps $V$ and $\Lambda$ are deterministic or probabilistic.

In all cases, the key ingredient of the certification is the observation of quantum nonlocality via Bell inequality violation \cite{Bell,review}. More precisely, we estimate the Clauser-Horne-Shimony-Holt \cite{CHSH} scores  
\be \label{eq: CHSH scores}
S_{i/o} \equiv \sum_{a,b,x,y=0,1} (-1)^{a+b +xy}\, \text{P}^{(i/o)}(a,b|x,y). 
\ee
In general, two different CHSH scores  are evaluated, when the memory is present (case ``o’’)  and when it is bypassed (case ``i’’). These CHSH scores allow to directly lower bounds the fidelities of the entangled states $\rho_{AB}^{(i/o)}$ with respect to the ideal two-qubit Bell state. Following the noise-robust self-testing results of Ref.~\cite{jed}, the state fidelities with and without the memory are bounded by:
\be\label{eq:jed formula}
F_{i/o} \geq f(S_{i/o}), \, \text{with} \quad f(S)\equiv\frac{1}{2}\left(1+ \frac{S-s_*}{2\sqrt{2}-s_*}\right)
\ee
where $s_* = \frac{16+14\sqrt{2}}{17}\approx 2.11$. We now show how these bounds on the fidelity can by used to certify the memory channel.

\textit{Scenario 1 ---} We first consider the case of a deterministic quantum memory, i.e.~one with success probability
\be
P_o=\tr(\widetilde{\t{QM}}[\rho^{(i)}])
\ee
equal to $P_o=1$. In this case, we prove the following lower bound on the fidelity of the quantum memory as a function of the two singlet fidelities $F_{i/o}$:
\be\label{eq: F bound ideal}
\cF_{\t{id}}^{(\Lambda,V)}(\widetilde{\t{QM}}) \geq 
\begin{cases}
\frac{F_o+\sqrt{2F_o-1}}{2} & \lambda_i\geq\frac{1}{2F_o}\\
\left(\frac{\sqrt{2F_o}-(\sqrt{\lambda_i}-\sqrt{1-\lambda_i}) }{2\sqrt{1-\lambda_i}}\right)^2 & \text{otherwise}
\end{cases}
\ee
with $\lambda_i = 
\frac{1}{2}+\sqrt{F_i(1-F_i)}$ if $ F_i\geq \frac{1}{2}$ and $\lambda_i = 1$ if $F_i <\frac{1}{2}$. The proof is detailed in Appendix~\ref{sec:scenario1}. It consists in proving the existence of deterministic injection and extraction maps, given by CPTP maps $V$ and $\Lambda$ respectively, achieving the above fidelity. Since the maps are trace-preserving, this bound provides a robust self-testing of the quantum memory with $P_\checkmark^{(\Lambda,V)}=1$.

Remarkably, we also prove that when the certified memory fidelity is above 50\%, i.e.~the memory is shown to preserve some entanglement, the lower-bound given in Eq.~\eqref{eq: F bound ideal} is tight. That is, there is a realization attaining the fidelities $F_{i/o}$, such that $\max_{\Lambda,V}\cF_{\t{id}}^{(\Lambda,V)}(\widetilde{\t{QM}})$ is given by the right hand side of Eq.~\eqref{eq: F bound ideal}.

The fidelity bound is illustrated in Fig.~\ref{fig:newBound} as a function of the two observed CHSH values $S_{i/o}$ through Eq.~\eqref{eq:jed formula}, showing substantial improvement over the previously known bound~\cite{blocks}. Notably, in contrast to the previous bound, the new one is asymmetric with respect to the two CHSH scores so that a high memory fidelity can still be certified when the CHSH value without memory is low, even below the local bound or unknown. In the later case one simply sets $\lambda_i=1$.

\begin{figure}
   \centering
    \includegraphics[width=0.48\textwidth]{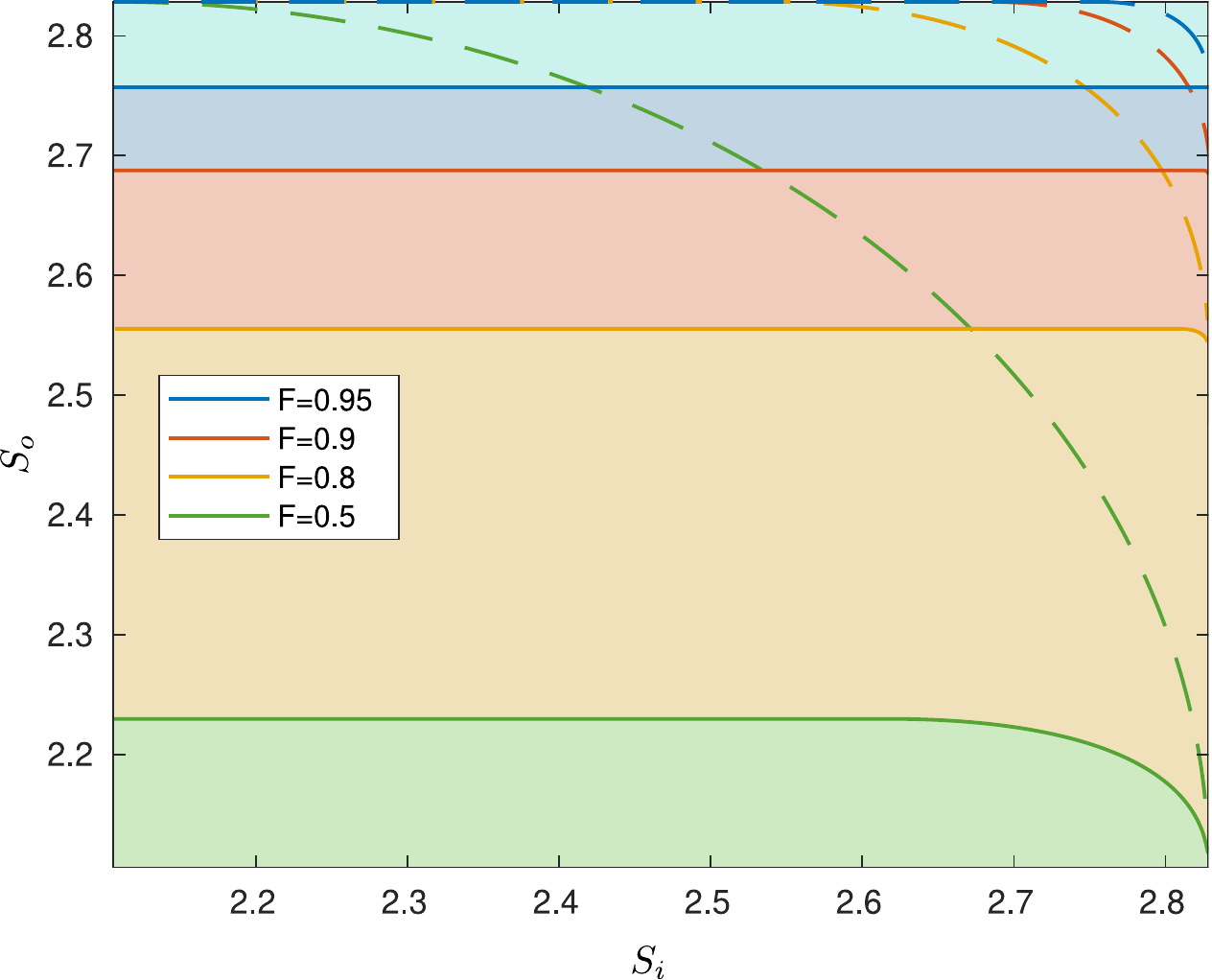}
    \caption{Fidelity of a quantum memory in absence of losses as a function of the initial and final CHSH values $S_i$ and $S_o$. The bound from~\cite{blocks} is shown in dashed lines for comparison.}
    \label{fig:newBound}
\end{figure}

\textit{Scenario 2 ---} Second, we consider the task of assessing the quality of a heralded probabilistic memory with $P_o \leq 1$ from the result of measurements performed after the storage only. In this case, we prove the following lower bounds on the quality of the quantum memory:
\be\label{eq: bound lossy}
\begin{split}
    \cF_{\t{id}}^{(\Lambda,\bar V)}(\widetilde{\t{QM}}) &\geq F_o \\
     P_{\checkmark}^{( \Lambda,\bar V)}(\widetilde{\t{QM}}) &\geq \frac{1}{2} P_o
\end{split}
\ee
for some CPTP map $\Lambda$ and CPTN map $\bar V$.

The full proof can be found in Appendix~\ref{sec:scenario2}. The result works by proving the existence of injection and extraction maps, the former is here probabilistic. Hence, the quantum state to be stored has to be filtered \textit{before} the memory by the injection map $\bar V$, with the filtering probability guaranteed to exceed 50\%. The necessity of such filtering can be understood intuitively from the realisation where the source prepares a non-maximally entangled two-qubit state $\ket{\Phi_{\lambda}} = \sqrt{\lambda}\ket{00} + \sqrt{1-\lambda}\ket{11}$ and the memory device filters it (e.g.~due to polarization dependent losses) into a maximally entangled state $\ket{\Phi^+}$, leading to a maximal CHSH score upon heralding of a successful memory storage.

\textit{Scenario 3 ---} Finally, we point out that it is also possible to assess the quality of a heralded probabilistic memory with both fidelities $F_i$ and $F_o$ (CHSH scores $S_i$ and $S_o$), and for a deterministic map $V$ and a probabilistic one $\bar \Lambda$. This is well adapted to an experiment where the CHSH score with bypassed memory $S_i$ is large. Here we derive lower bounds on $\cF_{\t{id}}^{(\bar \Lambda,V)}(\widetilde{\t{QM}})$ and $P_{\checkmark}^{( \bar \Lambda,V)}(\widetilde{\t{QM}})$ as results of a simple numerical minimization that has to be solved given the values $F_i$, $F_o$ and $P_o$. The details can be found in Appendix~\ref{sec:scenario3}.

\section{Calibration-independent certification}

In most current experiments particle losses such as finite detection efficiency of photonic detectors make black-box certification challenging. Furthermore, most of current quantum memories have a low efficiency but are not heralded, i.e.~one does not know if the memory is successful before measuring its output state. We now show how the proposed technique can be adapted to such cases via the introduction of fair-sampling assumptions on the measurement model. As discussed in~\cite{MatthiasBock}, the obtained certification retains a high degree of resistance to experimental errors.

In presence of detection inefficiency, the description of the binary measurement on system $A$ is modified to a 3-outcome POVM $\cM_{A}=\{A_{0|x},A_{1|x}, A_{\nc|x} \}$, and similarly for $\cM_{B}^{(i/o)}$. Here, the operators $A_{\nc|x}$ and $B_{\nc|y}^{(i/o)}$ are associated to the no-click events. The a priori quantum model of the setup then involves the modified elements $\{ \rho_{AB}^{(i)},\rho_{AB}^{(o)}, \cM_A, \cM_B^{(i)},\cM_B^{(o)}\}$ defined on the Hilbert spaces $\cH_A, \cH_B^{(i)}$ and $\cH_B^{(o)}$, yielding probabilities of the form of Eq.~\eqref{eq: P} with the larger alphabet $a,b=0,1,\nc$.

In this scenario, we define the post-selected correlations 
\be
\bar {\text{P}}^{(i/o)}(a,b|x,y) = \frac{\text{P}^{(i/o)}(a,b|x,y)}{\sum_{a,b=0,1}\text{P}^{(i/o)}(a,b|x,y)}
\ee
for $a,b =0,1$, estimated by simply discarding  all the events with $a=\nc $ or $b=\nc$. Post-selection opens the infamous detection loophole: when the occurrence of no-click events depends on the measurement setting, nonlocal correlation in the post-selected data can be faked by a local model~\cite{loophole1}. This possibility can be ruled out by introducing minimal physics-informed assumptions about the measurement apparatus. Following~\cite{Orsucci2020howpostselection} we introduce two assumptions of this type:
\begin{itemize}
\item[(SFS)] The measurement apparatus satisfies the strong fair-sampling assumption  if in the experiment the occurrence of no-click events is independent of the measured quantum state.
\item[(WFS)] The measurement apparatus satisfies the weak fair-sampling assumption if in the experiment the occurrence of no-click events is independent of the measurement setting.
\end{itemize}
We simply say that a measurement is SFS or WFS to mean that it satisfies the corresponding assumption. A more formal discussion of the assumptions and their implication for certification can be found in Appendix~\ref{app: assumptions}. Here we only give a brief recipe of how to use them together with the bounds derived in the previous section.

If all measurements are SFS the post-selection is essentially for free. In this case one computes the post-selected CHSH scores $\bar{S}_{i/o}$ from $\bar {\text{P}}^{(i/o)}(a,b|x,y)$ via Eq.~\eqref{eq: CHSH scores}. These values can be readily used to bound the singlet fidelities of the states $\rho_{AB}^{(i/o)}$ via Eq.~\eqref{eq:jed formula}, and the memory fidelity in Scenarios 1-3. Note however that in most experiments SFS is hard to justify for the measurement performed on the memory output. Indeed, if the photon can get lost inside the memory leading to a no-click outcome, the occurrence of the latter is manifestly not independent of the state that comes out of the memory.

If a measurement $\cM$ is WFS it can be decomposed as a filter $R$, which may produce a no-click outcome $\nc$, followed by an efficient measurement $\bar \cM$ which only outputs  $0$ or $1$. Thus, if a state $\rho_{AB}^{(i/o)}$ and inefficient WFS measurements $\cM_A$ and $\cM_{B}^{(i/o)}$ give rise to statistics $\text{P}^{(i/o)}(a,b|x,y)$, the post-selected statistics $\bar {\text{P}}^{(i/o)}(a,b|x,y)$ can be associated to the filtered state $\bar \rho_{AB}^{(i/o)} = \frac{R_A\otimes R_B^{(i/o)}[\rho_{AB}^{i/o}]}{R_A\otimes R_B^{(i/o)}[\rho_{AB}^{i/o}]}$ and to measurements $\bar \cM_A$ and $\bar \cM_B^{(i/o)}$ with unit efficiency. This allows one to use scenarios 2 and 3 in order to certify the probabilistic memory channel $R_B^{(o)}\circ \widetilde{\t{QM}}$. It is more elegant to absorb the filter $R_B^{(o)}$ in the extraction map, and view this task as certification of the memory channel $\widetilde{\t{QM}}$ with a probabilistic map $\bar \Lambda$.
In particular, for scenario 2 we then find 
\be\begin{split}
\cF_{\t{id}}^{(\bar \Lambda,\bar V)}(\widetilde{\t{QM}})&\geq f(\bar S_o)\\
P_{\checkmark}^{( \bar \Lambda,\bar V)}(\widetilde{\t{QM}}) &\geq \frac{1}{2} \t{P}^{(o)}(b\neq \nc| a\neq \nc).
\end{split}
\ee
Here the extraction map $\bar \Lambda$ is allowed to filter out the loss branch of the memory. This approach is thus suitable to certify non-heralded low efficiency memories, as we now illustrate.

\section{Experimental illustration}

To illustrate the practical relevance of our results, we apply our certification methods to the experimental data reported in~\cite{tiranov2015}. There half of a polarization and energy-time hyper-entangled photonic state is stored in an atomic ensemble quantum memory. The highest CHSH scores are reported for the energy-time degree of freedom, reaching $\bar S_i =2.733$ and $\bar S_o=2.64$ \footnote{For the polarisation degree of freedom, slightly lower CHSH values have been observed}. Note that these CHSH values are obtained for the post-selected data, as the setup features only low transmission, in the range of $1 \%$.

The first step consists in obtaining lower bounds on the fidelities of the entangled states with respect to a Bell states. More precisely, from Eq.~\eqref{eq:jed formula} we obtain the bounds on the fidelities of the filtered extracted states $F_i\geq0.93$ and $F_o \geq0.87$. 

For this experiment, where transmissions are relatively low, the most adapted approach is that of Scenario 2 together with the WFS assumption. Indeed the WFS assumption can be well justified for the measurement setup~\cite{Orsucci2020howpostselection}. From the results above, we can thus certify that the experimental device works as an announced qubit quantum memory with a fidelity lower bounded by
\be
\cF_{\t{id}}^{(\bar \Lambda,\bar V)}(\widetilde{\t{QM}}) \geq 0.87
\ee
The available experimental data did not allow us to estimate the conditional detection probability $\t{P}^{(o)}(b\neq \nc| a\neq \nc)$ as only the post-selected data is reported, hence we could not derive here a lower bound on the success probability of the quantum memory.

Note that in principle one could also apply the approach of Scenario 1 (or Scenario 3) to this experimental data. This would require making the WFS assumption on Alice, but also the SFS assumption on Bob, at least in the case where the memory is bypassed (Scenario 3). While the first assumption can be justified, the second is more delicate here. Indeed the SFS assumption requires that the state received by the measurement has no influence on the occurrence of the no-click outcome. In our setup this is hard to justify because the photon-number statistics of the states is nondeterministic, and the number of incoming photons influences the click probability. In order to apply the approach of Scenario 1 to an experiment, the latter would need to have either detection efficiencies close to one or involve a deterministic source of entanglement.

\section{Conclusion }

We have developed rigorous and efficient certification and characterisation methods for quantum memories. We derived our results in a device-independent setting, where all quantum devices are treated as black boxes. To make our approach more directly applicable to current experiments, we also discussed rigorously the conclusions that could still be made when discarding inconclusive events with various forms of fair sampling assumptions. To illustrate the practical relevance of our methods, we applied them to the data of a recent experiment \cite{tiranov2015} certifying a quantum memory fidelity of 0.87, without requiring an a priori characterisation of the source of the measurement devices. This further demonstrates the practical relevance of device-independent methods, such as self-testing, for the certification of practical devices.

More generally our methods can be applied to any device that implements an identity quantum channel, for example an optical fibre or transmission line~\cite{Neves2023}, a frequency-converter~\cite{MatthiasBock} or a simple error-correcting code, and may thus find broader applications in the future.

\begin{acknowledgments}
The authors acknowledge financial support from the Swiss National Science Foundation (project $2000021\_192244/1$ and NCCR SwissMAP).
\end{acknowledgments}

\bibliography{ST}{}

\appendix
\section{Background}

\label{app: assumptions}

In this section we formally introduce the results that will be used to derived the bounds presented in the main text.

\subsection{Fair sampling assumptions}

Here we describe general properties of the systems which satisfy one of the fair-sampling assumptions mentioned in the main text. These assumptions concern the form of the measurements $\cM_A, \cM_B^{(i)}$ and $\cM_B^{(o)}$, and allow one to describe the physics of the experiment in terms of the post-selected statistics
\be
\bar {\text{P}}^{(i/o)}(a,b|x,y) = \frac{\text{P}^{(i/o)}(a,b|x,y)}{\sum_{a,b=0,1}\text{P}^{(i/o)}(a,b|x,y)},
\ee
where no-click events ``$\nc$'' are discarded from the measurement data, as we now explain.

The \textit{strong fair sampling assumptions} (SFS) says that the occurrence of the no-click events is independent of the measured state. In this case the post-selection is ``for free'', as the  post-selected statistics $\bar {\text{P}}^{(i/o)}(a,b|x,y)$ correspond to the model
\be
\{ \rho_{AB}^{(i)}, \rho_{AB}^{(o)}, \bar{\cM}_A, \bar{\cM}_B^{(i)}, \bar{\cM}_B^{(o)}\},
\ee
where the measurements have been replaced with their unit-efficiency version $\bar{\cM}_A, \bar{\cM}_B^{(i)}, \bar{\cM}_B^{(o)}$ -- that always give a click ($\bar A_{\nc|x}= \bar B_{\nc|y}^{(i/o)}=0$). It is easy to see that the SFS assumption guarantees the existence of such measurements.

The \textit{weak fair sampling assumption} (WFS) is satisfied when the occurrence of the no-click events is independent of the measurement setting. In this case each measurement can be decomposed as a filter $R$, followed by a unit-efficiency measurement. The post-selected correlations can then be attributed to the model
\be\label{eq: model WFS}
\{ \bar{\rho}_{AB}^{(i)}, \bar{\rho}_{AB}^{(o)}, \bar{\cM}_A, \bar{\cM}_B^{(i)}, \bar{\cM}_B^{(o)}\},
\ee
where the state are filtered from $\rho_{AB}^{(i/o)}$ as
\be\label{eq: filtered state 1}
\bar{\rho}_{AB}^{(i/o)} =\frac{(R_A\otimes R_B^{(i/o)})[\rho_{AB}^{(i/o)}]}{\text{tr} (R_A\otimes R_B^{(i/o)})[\rho_{AB}^{(i/o)}]},
\ee
with some CPTN maps $R_A$ and $R_B^{(i/o)}$. Furthermore, the probability that these filters are successful can be estimated from the experimental data:
\be\label{eq app: prob 1}\begin{split}
\text{tr} (R_A\otimes R_B^{(i/o)})[\rho_{AB}^{(i/o)}] &= P^{(i/0)}(a\neq \nc,b\neq \nc|x,y)
\\
\text{tr} (R_A\otimes 
\id)[\rho_{AB}^{(i/o)}]&= P^{(i/0)}(a\neq \nc|x)
\end{split}
\ee
Note here that in general the probabilities $P^{(i/0)}(a\neq \nc,b\neq \nc|x,y)$, $P^{(i/0)}(a\neq \nc|x)$ and $P^{(i/0)}(b\neq \nc|y)$ may depend on the settings $x,y$. This is not in contradiction with the most general form of the WFS assumptions, where the filter maps for different setting are not identical but proportional~\cite{Orsucci2020howpostselection}. In this case, one would simply use Eq.~\eqref{eq app: prob 1} for any choice of $x$ and $y$.

Under the WFS assumption the post-selected CHSH scores $\bar S_{i/o}$ in Eq.~\eqref{eq: CHSH scores} can thus be associated to the filtered states $\bar{\rho}_{AB}^{(i/o)}$ and the measurements $\bar{\cM}_A$ with  $\bar{\cM}_B^{(i/o)}$. Furthermore, the probabilities to pass the filters are equal to the probabilities to get conclusive measurement outcomes, and can be estimated from the experimental data. 

Naturally, one can also consider a combination of the two types of fair-sampling assumptions. E.g., if Alice's  measurement is WFS and Bob's are SFS, the states $\bar{\rho}_{AB}^{(i/o)}$ are only filtered on Alice's side with the map CPTN map $R_A$. In the following, we consider the most general case where all measurements are WFS. In this case, one focuses on the states $\bar \rho_{AB}^{(i/o)}$ and measurements $\bar \cM_A, \bar \cM_{B}^{(i/o)}$ of Eqs.~(\ref{eq: model WFS},\ref{eq: filtered state 1}) that give rise to the CHSH scores $\bar{S}_{i/o}$. This situation includes both the case where some measurements are SFS and where no post-selection occurs: these special cases are recovered by setting $R_B^{(i/o)}=\text{id}$ and/or $R_A=\text{id}$ (in this case $\bar S_{i/o}=S_{i/o}$).

\subsection{CHSH as a self-test}

Let us recall that CHSH is a self-test for the maximally two-qubit state. Precisely, Kaniewski~\cite{jed} showed that if a state $\bar{\rho}^{(i/o)}_{AB}$ and measurements $\bar{\cM}_A$ and $\bar{\cM}_B$ achieve a CHSH score of $\bar S_{i/o}$, then there exist local CPTP maps $\Lambda_A$ and $\Lambda_B^{(i/o)}$ such that
\be\label{eq: JED}\begin{split}
&\bra{\Phi^+}(\Lambda_A\otimes\Lambda_B^{(i/o)})[\bar{\rho}^{(i/o)}_{AB}] \ket{\Phi^+} \geq f(\bar S_{i/o}) \\
&\, \text{with}\qquad f(S)=\frac{1}{2}\left(1+ \frac{S-s_*}{2\sqrt{2}-s_*}\right)
\end{split}
\ee   
and $s_* = \frac{16+14\sqrt{2}}{17}\leq 2.11$. Furthermore the form of the extraction map $\Lambda_A (\Lambda_B)$ only depend on the local measurements $\bar{\cM}_A (\bar{\cM}_B)$.

\subsection{Self-testing under fair-sampling assumption}

Let us now discuss in greater detail what the post-selected CHSH scores tell us about the underlying states.

Alice's measurement $\cM_A$ used for both CHSH tests (with or without memory) is the same. Hence, the filter $R_A$ in Eq.~\eqref{eq: filtered state 1} and the extraction map $\Lambda_A$ in Eq.~\eqref{eq: JED} are also common to both tests. This allows to define a state which absorbs these maps
\be
\varrho_{AB}= \frac{(\Lambda_A\!\circ\! R_A\otimes \t{id})[\rho_{AB}^{(i)}]}{\t{tr} (R_A\otimes \t{id})[\rho_{AB}^{(i)}]}
\in B(\mathds{C}^2_A\otimes \cH^{(i)}_B),
\ee
where the system $A$ is a qubit, and use it to rewrite the fidelities of Eq.~\eqref{eq: JED} as
\begin{align}\label{eq: Fi}
F_i &= \underbrace{\bra{\Phi^+} \frac{(\t{id}\otimes \bar{\Lambda}_B^{(i)})[\varrho_{AB}]}{\tr (\t{id}\otimes \bar{\Lambda}_B^{(i)})[\varrho_{AB}]}\ket{\Phi^+}}_{=\bra{\Phi^+}(\Lambda_A\otimes\Lambda_B^{(i)})[\bar{\rho}^{(i)}_{AB}] \ket{\Phi^+}}  %= F_i\geq f_i\equiv f(\bar S_i) 
\\  \label{eq: Fo}
F_o &= \underbrace{\bra{\Phi^+}\frac{(\t{id}\otimes \bar{\Lambda}_B^{(o)}\circ \widetilde{\t{QM}})[\varrho_{AB}]}{\tr (\t{id}\otimes \bar{\Lambda}_B^{(o)}\circ \widetilde{\t{QM}})[\varrho_{AB}]} \ket{\Phi^+} }_{=\bra{\Phi^+}(\Lambda_A\otimes\Lambda_B^{(o)})[\bar{\rho}^{(o)}_{AB}] \ket{\Phi^+}}%= F_o\geq f_o\equiv f(\bar S_o),
\end{align}
where we defined the CPTN maps $\bar{\Lambda}_B^{(i/0)} \equiv \Lambda_B^{(i/o)} \circ R_B^{(i/o)} :B(\cH_B^{(i/o)})\mapsto B(\mathds{C}^2)$ combining the filtering with the extraction of the qubit. Here and in the remaining of the text, we use capital letters ($V$) to denote CPTP maps and capital letters with a bar ($\bar V$) to denote CPTN maps. One also has access to the corresponding filtering probabilities
\begin{align} \label{eq: Pi}
P_i &= \tr (\t{id}\otimes \bar{\Lambda}_B^{(i)})[\varrho_{AB}] \\
\label{eq: Po}
 P_o &= \tr ( \t{id}\otimes \bar{\Lambda}_B^{(o)}\circ \widetilde{\t{QM}})[\varrho_{AB}] 
\end{align}
that are estimated from the experimental data. In the absence of no-detection events, $P_i=P_o=1$. Assuming strong fair sampling also amounts to set $P_i=1$ and $P_o=1$, as a consequence of setting $R_B^{(i/o)}=\text{id}$, which can be done independently for the measurements $\cM_B^{(i)}$ and $\cM_B^{(o)}$. Let us summarize this with the following observation.\\

\textbf{Observation} \textit{In all the scenarios discussed in the paper we can identify a state $\varrho_{AB}$ on $B(\mathds{C}^2_A\otimes \cH^{(i)}_B)$ and two maps $\bar \Lambda_B^{(i)}$ and $\bar \Lambda_B^{(o)}$, such that the bounds 
\be\label{eq:boundsFfPp}
    F_i \geq f_i\ ,\ \ F_o \geq f_o\ ,\ \ P_i \geq p_i\ ,\ \ P_o \geq p_o
\ee
hold for Eqs.~(\ref{eq: Fi}--\ref{eq: Po}) for some values of $f_{i}, f_{o}, p_{i}$ and $p_o$. In particular, we can always choose
\begin{align}
    p_i &= \t{P}^{(i)}(b\neq \nc| a\neq \nc)\\
    p_o &= \t{P}^{(o)}(b\neq \nc| a\neq \nc)
\end{align}
and Eq.~\eqref{eq: JED} provides the valid bounds
\begin{align}
    f_i &= f(\bar S_i)\\
    f_o &= f(\bar S_o).
    \label{eq app: obs final}
\end{align}}\\

This will be the starting point for the proof of the results.

\section{Proof of the main results.}

\label{app: proofs}
In this section we give the proofs of the main results.

\subsection{Preliminary Lemmas about $\varrho_{AB}$}

From the discussion above, one sees that the quantities $F_{i/o}\geq f_{i/o}$ and $P_{i/o}\geq p_{i/o}$ can be interpreted as statements about the state $\varrho_{AB}$ and the memory channel $\widetilde{\t{QM}}$. Yet there is still one obstacle that we have to overcome before combining Eqs.~(\ref{eq: Fi}-\ref{eq: Po}) in the final certification of the memory. The later requires to know how $\widetilde{\t{QM}}$ acts on the state $\Phi^+$, while the experimental data only tells us about  its action on the unknown state $\varrho_{AB}$. Our first goal is thus to show that $\varrho_{AB}$ can be prepared from $\Phi^+$ with the help of some injection map $V$ applied by Bob. Statements of this type are provided by the following Lemmas \ref{lemma0}-\ref{lemma2}. For Lemmas \ref{lemma0} and \ref{lemma2} we consider a different task of preparing $\varrho_{AB}$ from a non-maximally entangled two-qubit state. Its usefulness for our general goal will become clear in the next section.

 \begin{lemma} \label{lemma0} For any state $\varrho_{AB}\in B(\mathds{C}_A^2\otimes \cH_B^{(i)})$ there exists a qubit unitary $u$, a CPTP map $\Xi : B(\mathds{C}^2) \to B(\cH_B^{(i)})$ and a two qubit state 
\be\label{eq: Phi_p}
\ket{\Phi_\lambda}\equiv \sqrt{\lambda}\ket{0 0}+\sqrt{1-\lambda}\ket{11}
\ee
with some  $\frac{1}{2}\leq \lambda \leq 1$, such that
\be
\varrho_{AB} =(u\otimes \Xi)[\Phi_\lambda].
\ee
 \end{lemma}
 \begin{proof}
 The lemma is a simple consequence of the Schmidt decomposition. We start by purifying $\varrho_{AB}$. To do so introduce a system $B'$. There exist a pure state $\ket{\Psi}_{ABB'}$ which is a purification of $\varrho_{AB}$, that is  
\be
\varrho_{AB}= \tr_{B'} \ketbra{\Psi}_{ABB'}.
\ee
Next, put the states $\ket{\Psi}_{ABB'}$ in the Schmidt form with respect to the bipartition $A|BB'$
\be
\ket{\Psi}_{ABB'}= \sqrt{\lambda}\ket{a_0}_A\ket{\xi_0}_{BB'}+ \sqrt{1-\lambda}\ket{a_1}_A\ket{\xi_1}_{BB'},
\ee
where both pairs of state are orthogonal $\braket{a_0}{a_1}=\braket{\xi_0}{\xi_1}=0$, and we assume $\lambda \geq 1/2$ without loss of generality. The $\Psi_{ABB'}$ can be prepared form $\Phi_\lambda$ in Eq.~\eqref{eq: Phi_p} 
\be
\ket{\Psi}_{ABB'} = u\otimes U \ket{\Phi_\lambda}
\ee
by applying a unitary $u=\ketbra{a_0}{0} + \ketbra{a_1}{1}$ on system A and an isometry $U =\ketbra{\xi_0}{0}+\ketbra{\xi_1}{1}$ on system B.
Finally by defining the CPTP map
\be
\begin{split}
\Xi : B(\mathds{C}^2) &\to B(\cH_B^{(i)}) \\
\rho &\mapsto \tr_{B'} U[\rho]
\end{split}
\ee
we obtain
\be
\varrho_{AB} =(u\otimes \Xi)[\Phi_\lambda]
\ee
with $\frac{1}{2}\leq  \lambda \leq 1$. Here and below we slightly abuse the notation writing $u[\rho]= u \, \rho \, u^\dag$.
\end{proof}

Now consider the situation where $\varrho_{AB}$  is to be prepared from $\Phi^+$ with a non-deterministic injection map $\bar V$, that corresponds to Scenario 2.

 \begin{lemma} \label{lemma1} Any state $\varrho_{AB}\in B(\mathds{C}_A^2\otimes \cH_B^{(i)})$ can be prepared from $\Phi^+$ by a CPTN map $\bar V$ applied by Bob that has a success probability at least $\frac{1}{2}$, and a unitary $u$ applied by Alice. Precisely, for all $\varrho_{AB}$, there exist a CPTN map $\bar V$ and a unitary $u$ such that:
 \be
 \begin{split}
     &\varrho_{AB} = \frac{(u \otimes \bar V)[\Phi^+]}{\tr (u \otimes \bar V)[\Phi^+]} \\
     &\tr (u\otimes \bar V)[\Phi^+] \geq \frac{1}{2}.
 \end{split}
 \ee
 \end{lemma}
\begin{proof}
By Lemma~\ref{lemma0} the state $\varrho_{AB}$ can be expressed as
\be
\varrho_{AB} =(u\otimes \Xi)[\Phi_\lambda]
\ee
where $\Xi:B(\mathds{C}_A^2) \to B(\cH_B^{(i)})$ is a CPTP map and $u$ is a single-qubit unitary, and $\Phi_\lambda$ is a pure two-qubit state defined in Eq.~\eqref{eq: Phi_p}. In addition, any such $\Phi_\lambda$ can be obtained from $\Phi^+$ with a local filter $K_\lambda$
\be\begin{split}
    \Phi_\lambda = \frac{(\t{id}\otimes K_\lambda)[\Phi^+]}{\tr(\t{id}\otimes K_\lambda)[\Phi^+]} =\frac{(K_\lambda\otimes \t{id})[\Phi^+]}{\tr(K_\lambda\otimes \t{id})[\Phi^+]}\\
K_\lambda : \rho\mapsto
\left(\begin{array}{cc}
1 & \\
& \sqrt{\frac{1-\lambda}{\lambda}}
\end{array} \right)\rho \left(\begin{array}{cc}
1 & \\
& \sqrt{\frac{1-\lambda}{\lambda}}
\end{array} \right).
\end{split}
\ee
Here, the filtering probability $Q(p)=\tr(\t{id}\otimes K_\lambda)[\Phi^+] =\tr( K_\lambda\otimes \t{id})[\Phi^+]=\frac{1}{2}(1+ \frac{1-\lambda}{\lambda})$ takes values between $\frac{1}{2}\leq Q(\lambda) \leq 1$.

Hence, we obtain the following relations
\be
\varrho_{AB}= \frac{(u\otimes \Xi \circ K_\lambda)[\Phi^+]}{\t{tr} (u\otimes \Xi \circ K_\lambda)[\Phi^+]} =  \frac{(u\circ K_\lambda\otimes \Xi )[\Phi^+]}{\t{tr}(u\circ K_\lambda \otimes \Xi )[\Phi^+]},
\ee
specifying how the state $\rho_{AB}$ can be probabilistically prepared form $\Phi^+$ with local operations. In particular, let us define the injection map (CPTN) 
\be
\bar V = \Xi \circ K_\lambda  \circ u^T : B(\mathds{C}^2)\to H(\cH^{(i)}_B).
\ee
Using $(u\otimes\t{id})[\Phi^+]=(\t{id}\otimes u^T)[\Phi^+]$ we have 
\be
\frac{(\t{id}\otimes \bar V)[\Phi^+]}{\t{tr} (\t{id}\otimes \bar V)[\Phi^+]} = \frac{(u\otimes \Xi \circ K_\lambda )[\Phi^+]}{\t{tr} (u\otimes \Xi \circ K_\lambda )[\Phi^+]}=\varrho_{AB} .
\ee
In addition, the success probability satisfies
\be
\t{tr} (u \otimes \bar V)[\Phi^+]=  \t{tr} (\t{id}\otimes \bar V)[\Phi^+] = Q(\lambda)\geq \frac{1}{2},
\ee
which concludes the proof.
\end{proof}
%%%%%%%%%%%%%%%%%%%%%%%%%%%%%%%%%%%%%%%%%%%%%%%%%%%%
%%%%%%%%%%%%%%%%%%%%%%%%%%%%%%%%%%%%%%%%%%%%%%%%%%%%

It is not always desirable to have a probabilistic injection map, in particular scenarios 1 and 3 are stated for a deterministic one.
For this reason we now consider how to prepare $\varrho_{AB}$ from a state $\Phi_\lambda$ with the help of a deterministic injection map $V$. In contrast to Lemma~\ref{lemma0} we now also use the additional information 
provided by the bounds on $F_i$ and $P_i$ in order restrict the possible values of $\lambda$.

 \begin{lemma} \label{lemma2} Any state $\varrho_{AB}\in B(\mathds{C}_A^2\otimes \cH_B^{(i)})$ satisfying\be\begin{split}\label{eq:40}
\tr (\t{id}\otimes \bar{\Lambda}_B^{(i)} )[\varrho_{AB}] &= P_i \\
F\left(\frac{1}{P}_i (\t{id}\otimes \bar{\Lambda}_B^{(i)})[\varrho_{AB}], \Phi^+ \right)  &= F_i 
\end{split}
\ee 
for some CPTN map $\bar{\Lambda}_B^{(i)} :B(\cH_B^{(i)})\mapsto B(\mathds{C}^2)$ can be prepared with a CPTP map $V$ applied by Bob and a unitary $u$ applied by Alice as
 \be\label{app eq: LB p condition}
 \varrho_{AB} = (u \otimes V)[\Phi_\lambda]
 \ee
 \textit{for a two qubit state $\ket{\Phi_\lambda} =\sqrt{\lambda}\ket{00}+\sqrt{1-\lambda}\ket{11}$ with }
\be\label{app eq: LB p}
\frac{1}{2} \leq  \lambda \leq \lambda_i\equiv \begin{cases} 
1-\left(\frac{1}{2} -\sqrt{F_i(1-F_i)}\right) P_i & F_i\geq \frac{1}{2} \\
1 & F_i < \frac{1}{2}
\end{cases}.
\ee
In addition, this bound is tight, in the sense that there exists a pure two qubit state $\varrho_{AB}$ and a channel $\bar{\Lambda}_B^{(i)}$ yielding the values $P_i$ and $F_i$ in Eq.~\eqref{eq:40}, such that Eq.~\eqref{app eq: LB p condition} is impossible to satisfy for any $\lambda>\lambda_i$ in Eq.~\eqref{app eq: LB p}.
\end{lemma}

\begin{proof}
By Lemma \ref{lemma0} we know that the state $\varrho_{AB}$ is of the form
\be
\varrho_{AB}= (u\otimes \Xi)[\Phi_\lambda]
\ee
with $\ket{\Phi_\lambda}=\sqrt{\lambda}\ket{00}+\sqrt{1-\lambda}\ket{11}$ for some $\frac{1}{2} \leq \lambda\leq 1$. This proves the lemma for $F_i<\frac{1}{2}$. In the following we thus consider the more interesting case $F_i \geq \frac{1}{2}$.

Rewrite \Cref{eq:40} in the form
\begin{align}
\tr (u\otimes \bar{\Lambda}_B^{(i)} \circ \Xi)[\Phi_\lambda] = P_i \\
\bra{\Phi^+}  (u\otimes \bar{\Lambda}_B^{(i)} \circ \Xi)[\Phi_\lambda] \ket{\Phi^+} = F_i P_i .
\end{align}
Since $u\otimes \id \ket{\Phi^+}= \id \otimes u^T \ket{\Phi^+}$,  for $\bar \Lambda= u^T \circ \bar{\Lambda}_B^{(i)} \circ \Xi$ these equations are equivalent to 
\begin{align}
\tr (\t{id }\otimes \bar \Lambda)[\Phi_\lambda] = P_i \\
\bra{\Phi^+}  (\t{id}\otimes \bar \Lambda)[\Phi_\lambda] \ket{\Phi^+} = F_i P_i .
\end{align}
We are looking for  the maximal value of $\lambda\geq \frac{1}{2}$ compatible with these equations for a  CPTN-map $\bar \Lambda : B(\mathds{C}^2)\to B(\mathds{C}^2)$ of the form above. To do so define the operator 
\be
\sigma_{AB}\equiv(\t{id}\otimes \bar \Lambda)[\Phi_\lambda].
\ee
It has to be positive $\sigma_{AB} \succeq 0$. Furthermore, by construction the marginal state satisfies
\be
\sigma_A = \tr_B \sigma_{AB} \preceq \tr_B \Phi_\lambda =
\left( \begin{array}{cc}
\lambda & \\
& 1-\lambda
\end{array}\right),
\ee
expressed in the computational basis. The maximal value of $\lambda$ can thus be bounded by the following semidefinite program
\be\label{eq:sdp1}
\begin{split}
    \lambda_{max}(F,P) = \max_{\lambda,\sigma_{AB}} \, & \lambda \\
    \text{subject to } \, &\sigma_{AB} \succeq 0 \\
     & \tr_B \sigma_{AB} \preceq \left( \begin{array}{cc}
\lambda& \\
& 1-\lambda
\end{array}\right)
\\
& \tr \sigma_{AB} = P_i \\
& \tr \Phi^+ \sigma_{AB} = F_i P_i.
\end{split}
\ee

Define the operator $T=A\otimes\id_2 + c \id_4 + d\ketbra{\Phi^+}$ with $A=\left(\begin{array}{cc} 0 & 0 \\ 0 & 1\end{array}\right)$, $B=\left(\begin{array}{cc} 1 & 0 \\ 0 & -1\end{array}\right)$,  $c=\frac{F_i-1+\sqrt{F_i(1-F_i)}}{2(1-F_i)}$, and $d=\frac{1-2F_i}{2\sqrt{F_i(1-F_i)}}$. The smallest eigenvalue of $T$ is $\frac{F_i-1+\sqrt{F_i(1-F_i)}}{2(1-F_i)}$, which is positive whenever $F_i\geq 1/2$. Hence, in the case $F_i\geq 1/2$, we have $T\succeq 0$ and we can write for any $\sigma_{AB}\succeq 0$:
\begin{equation}
\begin{split}
    \lambda & \leq \lambda + \tr (\sigma_{AB} T)\\
    &= -\lambda \tr (B A) + \tr \left(\sigma_{AB}\left(A\otimes\id_2 + c\id_4 + d\ketbra{\Phi^+}\right)\right)\\
    &= \tr \left(\left(-\lambda B + \tr_B \left(\sigma_{AB}\right)\right) A\right) + c\tr (\sigma_{AB}) \\
    &+ d\bra{\Phi^+}\sigma_{AB}\ket{\Phi^+}\\
    &\leq \tr(A A) + c P_i + d F_i P_i,
\end{split}
\end{equation}
where we used the conditions in \eqref{eq:sdp1}, in particular the second constraint which can be rewritten as $-\lambda B+\tr_B \sigma_{AB}\succeq A$. We thus obtain the following upper bound on $p$:
\begin{equation}\label{ap: ub p}
    \lambda \leq 1 - \left(\frac12-\sqrt{F_i(1-F_i)}\right)P_i.
\end{equation}

 Finally, let us show that the bound is tight. First, let us assume that $F_i\geq \frac{1}{2}$ and consider two qubit initial  state $\varrho_{AB} =\Phi_\lambda$ with $ \lambda = 1 - \left(\frac12-\sqrt{F_i(1-F_i)}\right)P_i$ and a qubit to qubit CPTN map $\bar \Lambda_B^{(i)}$ with a single Kraus operator given by 
\be
\begin{split}
\bar  \Lambda_B^{(i)}[\rho] &= K \rho K^\dag
\qquad \text{with} \qquad K =
\left(
\begin{array}{cc}
\sqrt{X} & \\
& 1
\end{array} \right),
\\X &= { \scriptstyle \frac{4 F_i^2 P_i^2-4 F_i P_i^2+4 \sqrt{F_i-F_i^2} \left(P_i-P_i^2\right)-P_i^2+2 P_i}{4 F_i^2 P_i^2-4 F_i P_i^2+P_i^2-4 P_i+4}}.
\end{split}
\ee
It is straightforward to see that this combination of the state $\varrho_{AB}$ and the channel $\bar  \Lambda_B^{(i)}$ precisely yields the values
\be\begin{split}
 &\tr  (\id\otimes \bar  \Lambda_B^{(i)})[\rho_{AB}]  = 1- \lambda(1-X) =P_i,\\
&\frac{\bra{\Phi^+}(\id\otimes \bar  \Lambda_B^{(i)})[\rho_{AB}]\ket{\Phi^+}}{P_i}= \frac{\left(\sqrt{\lambda} \sqrt{X}+\sqrt{1-\lambda}\right)^2}{2 (1- \lambda (1-X))} =F_i.
\end{split}
\ee
Hence, the values $F_i\geq \frac{1}{2}$ and $P_i$ are compatible with an initial state $\Phi_\lambda$ with $\lambda = 1 - \left(\frac12-\sqrt{F_i(1-F_i)}\right)P_i \geq \frac{1}{2}$ saturating the bound \eqref{app eq: LB p}. Furthermore, it is clear the state $\varrho_{AB} =\Phi_\lambda$ can not be locally prepared
\be
\Phi_\lambda \neq (u \otimes V)[\Phi_{\lambda'}]
\ee
with deterministic maps from a state $\Phi_{\lambda'}$ with $\lambda'>\lambda$ (e.g.~because $\Phi_\lambda$ is more entangled than $\Phi_\lambda'$ as revealed by an entanglement measure such as negativity). Therefore, for $F_i\geq 1/2$ the bound \eqref{app eq: LB p} is tight. 

For $F_i<\frac{1}{2}$, the value initial state to consider is $\varrho_{AB}= \ketbra{11}$. Taking a simple CPTN map $\bar  \Lambda_B^{(i)}[\rho] =  P_i\, U \rho U^\dag$, with a unitary $U$ satisfying $|\bra{1} U \ket{1}|^2= 2 F_i$, one sees that this state is compatible with any values $P_i, F_i<\frac{1}{2}$. Moreover, here the bound $\lambda\leq p_i=1$ is tight by definition, as there is no physical state $\ket{\Phi_\lambda}$ with $\lambda>1$.
\end{proof}
%%%%%%%%%%%%%%%%%%%%%%%%%%%%%%%%%%%%%%%%%%%%%%%%%%%%%%%%%
%%%%%%%%%%%%%%%%%%%%%%%%%%%%%%%%%%%%%%%%%%%%%%%%%%%%%%%%%

We now have all the tools we need to prove the results presented in the main text.  For all the three scenarios we start by giving a formal statement of the result, and then present a proof thereof. We proceed in order of increasing complexity. 

\subsection{Scenario 2} \label{sec:scenario2}

\noindent \textbf{Result 2.} \textit{ For any state $\varrho_{AB}\in B(\mathds{C}_A^2\otimes \cH_B^{(i)})$ and a CPTN map $\widetilde{\t{QM}}$ satisfying
\be\begin{split}
\t\tr ( \t{id}\otimes \bar{\Lambda}_B^{(o)}\circ \widetilde{\t{QM}})[\varrho_{AB}] &\geq p_o \\
\bra{\Phi^+}\frac{(\t{id}\otimes \bar{\Lambda}_B^{(o)}\circ \widetilde{\t{QM}})[\varrho_{AB}]}{\tr (\t{id}\otimes \bar{\Lambda}_B^{(o)}\circ \widetilde{\t{QM}})[\varrho_{AB}]} \ket{\Phi^+}  &\geq f_o
\end{split}
\ee 
for some CPTN map $\bar{\Lambda}_B^{(o)} :B(\cH_B^{(o})\mapsto B(\mathds{C}^2)$, there exist CPTN maps $\bar \Lambda$ and $\bar V$ such that 
\be\label{eq: res 3}
\begin{split}
    \cF_{\t{id}}^{(\bar \Lambda,\bar V)}(\widetilde{\t{QM}}) &\geq f_o \\
     P_{\checkmark}^{(\bar \Lambda,\bar V)}(\widetilde{\t{QM}})  &\geq \frac{1}{2} p_o.
\end{split}
\ee
In addition, if $\bar{\Lambda}_B^{(o)}$ is trace preserving, $\bar \Lambda$ is also trace preserving.} 

\begin{proof}
By Lemma~\ref{lemma1} there also exists a CPTN map $\bar V$ such that
\be
 \begin{split}
     &\varrho_{AB} = \frac{(u \otimes \bar V)[\Phi^+]}{\tr (u\otimes \bar V)[\Phi^+]} \\
     &\tr (u\otimes \bar V)[\Phi^+] \geq \frac{1}{2}.
 \end{split}
 \ee
It remains to combine the two results. First for any CPTN extraction map $\bar \Lambda$ we have
\be\begin{split}
  \cF_{\t{id}}^{(\bar \Lambda,\bar V)}(\widetilde{\t{QM}})  
  &=  \bra{\Phi^+}\frac{(\t{id}\otimes \bar{\Lambda}\circ \widetilde{\t{QM}}\circ \bar V)[\Phi^+]}{\tr (\t{id}\otimes \bar{\Lambda}\circ \widetilde{\t{QM}}\circ \bar V)[\Phi^+]} \ket{\Phi^+}\\
  &=  \bra{\Phi^+}\frac{(u^\dag u \otimes \bar{\Lambda}\circ \widetilde{\t{QM}}\circ \bar V)[\Phi^+]}{\tr (\t{id}\otimes \bar{\Lambda}\circ \widetilde{\t{QM}}\circ \bar V)[\Phi^+]} \ket{\Phi^+}\\
  &=\bra{\Phi^+}\frac{(u^\dag \otimes \bar{\Lambda}\circ \widetilde{\t{QM}})[\varrho_{AB}]}{\tr (\t{id}\otimes \bar{\Lambda}\circ \widetilde{\t{QM}})[\varrho_{AB}]} \ket{\Phi^+}.
\end{split}
\ee
Now by $u \otimes \mathds{1} \ket{\Phi^+} = \mathds{1}\otimes u^T \ket{\Phi^+}$ we have 
\be
\bra{\Phi^+} (u^\dag\otimes \t{id})[ \rho] \ket{\Phi^+} = 
\bra{\Phi^+} (\t{id }\otimes u^*)[ \rho] \ket{\Phi^+}
\ee
and thus
\be
  \cF_{\t{id}}^{(\bar \Lambda,\bar V)}(\widetilde{\t{QM}})  
  =\bra{\Phi^+}\frac{(\t{id} \otimes u^*\circ \bar{\Lambda}\circ \widetilde{\t{QM}})[\varrho_{AB}]}{\tr (\t{id}\otimes u^* \circ \bar{\Lambda}\circ \widetilde{\t{QM}})[\varrho_{AB}]} \ket{\Phi^+}\\
\ee
where we used the invariance of the trace under unitaries. Finally,
by setting $\bar \Lambda \equiv u^T \circ \bar{\Lambda}_B^{(o)}$ we find
\be\begin{split}
  \cF_{\t{id}}^{(\bar \Lambda,\bar V)}(\widetilde{\t{QM}})  
  &=\bra{\Phi^+}\frac{(\t{id} \otimes u^* u ^T \circ \bar{\Lambda}_B^{(o)}\circ \widetilde{\t{QM}})[\varrho_{AB}]}{\tr (\t{id}\otimes  \bar{\Lambda}_B^{(o)}\circ \widetilde{\t{QM}})[\varrho_{AB}]} \ket{\Phi^+}\\
  &=\bra{\Phi^+}\frac{(\t{id} \otimes \bar{\Lambda}_B^{(o)}\circ \widetilde{\t{QM}})[\varrho_{AB}]}{\tr (\t{id}\otimes  \bar{\Lambda}_B^{(o)}\circ \widetilde{\t{QM}})[\varrho_{AB}]} \ket{\Phi^+}\\
  &\geq f_o. 
\end{split}
\ee
For the success probability we have 
\be\begin{split}
 P_{\checkmark}^{(\bar \Lambda,\bar V)}(\widetilde{\t{QM}}) &=
 \tr (\t{id}\otimes \bar{\Lambda}\circ \widetilde{\t{QM}}\circ \bar V)[\Phi^+] \\
    &\geq \frac{1}{2} \tr ( \t{id}\otimes \bar{\Lambda}_B^{(o)}\circ \widetilde{\t{QM}})[\varrho_{AB}] \\
    &\geq \frac{1}{2} p_o.
\end{split}
\ee
Finally note that if $\bar{\Lambda}_B^{(o)}$, the extraction map $\bar \Lambda \equiv u^T \circ \bar{\Lambda}_B^{(o)}$ is also trace preserving by construction, concluding the proof.
\end{proof}

\subsection{Scenario 1} \label{sec:scenario1}

In the case of scenario 1 all the considered maps are deterministic. Before stating the result let us start by expression our goal more formally. For any instance  of $\varrho_{AB}, \widetilde{\t{QM}}$,  $\Lambda_B^{(i)}$ and $\Lambda_B^{(o)}$ satisfying $F_i\geq f_i$ and $F_o\geq f_o$, we wish to find the injection/extraction maps $\Lambda, V$ which maximize the Choi fidelity of the memory channel $\widetilde{\t{QM}}$. Hence, our goal can be expressed as the following optimization problem
\be
\min_{\substack{\varrho_{AB}, \widetilde{\t{QM}},  \Lambda_B^{(i)},\Lambda_B^{(o)}\\
\text{such that}\,\, F_{i/o}\geq f_{i/o}}}  \max_{\Lambda,\bar V} \,\cF_{\t{id}}^{(\Lambda, V)}(\widetilde{\t{QM}}).
\ee
Solving such a min-max problem is a priori not straightforward. Nevertheless, we can lower bound this quantity by considering an explicit construction of the maps $\Lambda$ and $V$ (Result 1). In addition, we can also upper bound this quantity by considering a particular realization $\{\varrho_{AB}, \widetilde{\t{QM}}, \bar \Lambda_B^{(i)},\bar \Lambda_B^{(o)}\}$ (Result 1'). Remarkably, the lower and upper bounds coincide in the parameter region of interest, showing that our bound is tight in this regime.\\

\noindent \textbf{Result 1.} \textit{ Consider a state $\varrho_{AB}\in B(\mathds{C}_A^2\otimes \cH_B^{(i)})$ and CPTP maps $\widetilde{\t{QM}},$  $\Lambda_B^{(i)}$ and $\Lambda_B^{(o)}$, such that  $(\t{id}\otimes \Lambda_B^{(i)})[\varrho_{AB}]$ and  $(\t{id}\otimes \Lambda_B^{(o)}\circ \widetilde{\t{QM}})[\varrho_{AB}]$ are two qubit states satisfying}
\begin{align}\label{eq: det Fi}
\bra{\Phi^+} (\t{id}\otimes \Lambda_B^{(i)})[\varrho_{AB}]\ket{\Phi^+}  &\geq f_i
\\  \label{eq: det Fo}
\bra{\Phi^+}(\t{id}\otimes \Lambda_B^{(o)}\circ \widetilde{\t{QM}})[\varrho_{AB}] \ket{\Phi^+} &\geq f_o\geq \frac{1}{2}.
\end{align} \textit{
There exist CPTP maps $\Lambda$ and $V$ such that }
\be\label{app eq: result 1 bound}
\cF_{\t{id}}^{(\Lambda,V)}(\widetilde{\t{QM}}) \geq 
\begin{cases}
\frac{f_o+\sqrt{2f_o-1}}{2} & \lambda_i\geq\frac{1}{2f_o}\\
\left(\frac{\sqrt{2f_o}-(\sqrt{\lambda_i}-\sqrt{1-\lambda_i}) }{2\sqrt{1-\lambda_i}}\right)^2 & \textit{otherwise}
\end{cases}
\ee
\textit{with} $\lambda_i \equiv \begin{cases}
\frac{1}{2}+\sqrt{f_i(1-f_i)} & f_i\geq \frac{1}{2} \\
1 & f_i <\frac{1}{2}\end{cases}$.
\begin{proof}
As a first step we want to exploit the inequality~\eqref{eq: det Fi} in order to learn something about the initial state $\varrho_{AB}$. Here we may use our Lemma~\ref{lemma2} for the particular case $P_i=1$. It guarantees the existence of a unitary $u$ and a CPTP map $V$ such that
\be
\begin{split}
&\varrho_{AB} = (u \otimes V)[\Phi_\lambda] \quad \text{with}
\\
&\frac{1}{2} \leq \lambda \leq \lambda_i \equiv
\begin{cases} \frac{1}{2}+\sqrt{f_i(1-f_i)} & f_i\geq \frac{1}{2} \\
1 & f_i< \frac{1}{2}
\end{cases}.
\end{split}\ee
We can thus rewrite Eq.~\eqref{eq: det Fo} with
\be\begin{split}
&\bra{\Phi^+}(\t{id}\otimes \Lambda_B^{(o)}\circ \widetilde{\t{QM}})[\varrho_{AB}] \ket{\Phi^+}\\
& =\bra{\Phi^+}(u \otimes \Lambda_B^{(o)}\circ \widetilde{\t{QM}}\circ V)[\Phi_\lambda] \ket{\Phi^+} \\
&= \bra{\Phi^+}(\t{id }\otimes u^T\circ \Lambda_B^{(o)}\circ \widetilde{\t{QM}}\circ V)[\Phi_\lambda] \ket{\Phi^+} \\
& = \bra{\Phi^+}(\t{id }\otimes \Lambda\circ \widetilde{\t{QM}} \circ V)[\Phi_\lambda] \ket{\Phi^+},
\end{split}
\ee
where we defined $\Lambda \equiv u^T\circ \Lambda_B^{(o)}$. At this point we have defined promising maps $\Lambda$ and $V$ defing the the qubit map $\cE \equiv \Lambda \circ \widetilde{\t{QM}}\circ V$, and established that the bounds of Eqs.~(\ref{eq: det Fi},\ref{eq: det Fo}) gurantee that  
\be\label{eq: r2 Fo}
\bra{\Phi^+}(\t{id }\otimes \cE) [\Phi_\lambda] \ket{\Phi^+} \geq f_o.
\ee
holds for some $\frac{1}{2}\leq \lambda \leq \lambda_i$. \\

Let us now assume a fixed value of $\lambda$ in this interval. We are interested in the  quantity
\be
\cF_{\t{id}}^{(\Lambda,V)}(\widetilde{\t{QM}}) = \bra{\Phi^+}(\t{id }\otimes \cE)[\Phi^+] \ket{\Phi^+},
\ee
which is different from left had side of \eqref{eq: r2 Fo}.  Nevertheless, the two quantities are not independent, and a lower bound on $\cF_{\t{id}}^{(\Lambda,V)}(\widetilde{\t{QM}})\geq G(f_o,\lambda)$ follows from the following optimization
\be\label{eq app: mini app}
\begin{split}
  G(f_o,\lambda)=  \min_{\cE} &\bra{\Phi^+}(\t{id }\otimes \cE)[\Phi^+] \ket{\Phi^+} \\
    \text{subject to} & \bra{\Phi^+}(\t{id }\otimes \cE) [\Phi_\lambda] \ket{\Phi^+} \geq f_o.
\end{split}
\ee
This minimization is in fact an SDP and can be solved efficiently. To see this define the Choi state $\sigma\equiv (\t{id}\otimes\cE)[\Phi^+]$ associated to the CPTP map $\cE$. By construction the Choi state satisfies
\be
\sigma_A = \tr_B \sigma = \frac{1}{2} \mathds{1}.
\ee
Furthermore, there is a one-to-one correspondence between positive semi-definite operators $\sigma\succeq 0$ satisfying this constraint and CPTP maps $\cE$. In terms of $\sigma$ our goal function~\eqref{eq app: mini app} simply reads 
\be
\cF_{\t{id}}^{(\Lambda,V)}(\widetilde{\t{QM}}) = \tr \sigma \Phi^+.
\ee
Next, note that the two state $\Phi_\lambda$ and $\Phi^+$ are related by 
\begin{align}\label{eq: smth1}
\ket{\Phi_\lambda} &=  T \ket{\Phi^+}\\ 
T &= \left(\begin{array}{cc} 
\sqrt{2\lambda} & \\
& \sqrt{2(1-\lambda)}
\end{array} \right)\otimes \id .
\end{align}
Hence, the state resulting from the action of $\text{id}\otimes \cE$ on $\Phi_\lambda$ reads $(\id \otimes \cE)[\Phi_\lambda] = T \sigma \, T$. Therefore, in term of $\sigma$ the constraint in Eq.~\eqref{eq app: mini app} reads
\be
\tr  T \sigma \, T \Phi^+ \geq f_o.
\ee
This allows us to rewrite Eq.~\eqref{eq app: mini app} in a form which is manifestly an SDP
\be\begin{split}\label{eq:sdp2}
G(f_o,\lambda) = \min_{\sigma}\, & \tr \sigma \, \Phi^+ \\
\t{s.t.}\, &\tr T \sigma T \, \Phi^+ \geq f_o\\
&\sigma_A = \tr_B \sigma = \frac{1}{2} \id \\
& \sigma \succeq 0.
\end{split}
\ee\\

It is however important to recall that we can only certify that $\lambda$ is below a certain value, but not above. So the final bound on the fidelity is given by
\be
\cF_{\t{id}}^{(\Lambda,V)}(\widetilde{\t{QM}}) \geq F'(f_o ,f_i) = \min_{\frac{1}{2}\leq \lambda\leq \lambda_i}   G(f_o,\lambda).
\ee
\bigskip
%%%%%%%%%%%%%%%%%%%%%%%%%%%%%%%%%%%%%%%%%%%%%%%%%%%%%

Let's now compute this function $F'$. First, focusing on $G(f_o,\lambda)$ defined in \eqref{eq:sdp2}, we consider the operator
\begin{equation}
    W=\Phi^+ - a \,T\Phi^+T + b \ketbra{0}\otimes \id.
\end{equation}
One can check that the eigenvalues of $W$ are positive whenever $\lambda\geq1/2$ and $f_o\geq 1/2$ with the following choice of parameters:
\begin{equation}
    \begin{split}
        a&=\frac{\sqrt{2f_o}-(\sqrt{\lambda}-\sqrt{1-\lambda})}{2\sqrt{2f_o}(1-p)}\\
        b&=\frac{\sqrt{2f_o} - (\sqrt{\lambda}-\sqrt{1-\lambda})}{2(1-\lambda)}(\sqrt{\lambda}-\sqrt{1-\lambda}).
    \end{split}
\end{equation}
Therefore, the objective function of \eqref{eq:sdp2} can be bounded as
\be\begin{split}
    \tr \sigma \Phi^+ &\geq a \tr \sigma T\Phi^+ T - b \tr \sigma  \ketbra{0}\otimes \id_2 \\
    & \geq a f_o - b \tr_A \sigma_A \ketbra{0} \\
    & = a f_o - \frac{b}{2}
    \\
    &= \left(\frac{\sqrt{2f_o} - (\sqrt{\lambda} - \sqrt{1-\lambda})}{2\sqrt{1-\lambda}}\right)^2,
\end{split}
\ee
where we used the fact that $a\geq 0$, which is valid when $f_o\geq1/2$. 

We shall now consider $G(f_o,\lambda)$ as a function of $p$ in the range  $\lambda\in \left[\frac{1}{2}, \lambda_i \right]$. We note that  $G(f_o,\lambda)$ has a unique local minimum at $\lambda=\frac{1}{2 f_o}$, yielding
\begin{equation}
    F'(f_o,f_i) = \frac{f_o+\sqrt{2f_o-1}}{2}
\end{equation}
whenever the local minimum is inside the interval $\frac{1}{2f_o}\in \left[\frac{1}{2}, \lambda_i \right]$. Here,  $\frac{1}{2}\leq \frac{1}{2f_0}\leq 1$ so
\be
\frac{1}{2f_o}\in \left[\frac{1}{2}, \lambda_i \right]
\Leftrightarrow \frac{1}{2f_o}\leq \lambda_i
\ee 

When this condition is not met, i.e.~the local minimum is on the right of the interval $\frac{1}{2f_o}> \lambda_i$, the function $G(f_o, \lambda_i)$ is decreasing for $\lambda\in \left[\frac{1}{2},\lambda_i\right]$. Hence its minimum is attained for $F'(f_o,f_i) = G(f_o, \lambda_i)$. Combining the two cases gives 
\be
F'(f_o,f_i) = 
\begin{cases}
\frac{f_o+\sqrt{2f_o-1}}{2} & \frac{1}{2f_o}\leq \lambda_i \\
\left(\frac{\sqrt{2f_o} - (\sqrt{\lambda_i} - \sqrt{1-\lambda_i})}{2\sqrt{1-\lambda_i}}\right)^2 & \frac{1}{2f_o}>  \lambda_i
\end{cases}.
\ee
This concludes the proof of the bound.
\end{proof}

Let us now show that the bound provided by the this result is tight whenever the certified memory fidelity exceeds 50\%. This follows from the following result.

\noindent \textbf{Result 1'.} \textit{ There exists a state $\varrho_{AB}\in B(\mathds{C}_A^2\otimes \cH_B^{(i)})$ and CPTP maps $\widetilde{\t{QM}},$  $\Lambda_B^{(i)}$ and $\Lambda_B^{(o)}$ achieving}
\begin{align}
\bra{\Phi^+} (\t{id}\otimes \Lambda_B^{(i)})[\varrho_{AB}]\ket{\Phi^+}  & =f_i
\\
\bra{\Phi^+}(\t{id}\otimes \Lambda_B^{(o)}\circ \widetilde{\t{QM}})[\varrho_{AB}] \ket{\Phi^+} & = f_o\geq \frac{1}{2} \label{eq app fo r1p}
\end{align} \textit{
such that  }
\be\label{app eq: result 1 bound dfadfa}
\max_{\Lambda,V}\cF_{\t{id}}^{(\Lambda,V)}(\widetilde{\t{QM}}) \leq 
\begin{cases}
\frac{f_o+\sqrt{2f_o-1}}{2} & \lambda_i\geq\frac{1}{2f_o}\\
\left(\frac{\sqrt{2f_o}-(\sqrt{\lambda_i}-\sqrt{1-\lambda_i}) }{2\sqrt{1-\lambda_i}}\right)^2 & \textit{otherwise}
\end{cases}
\ee
\textit{with $\lambda_i \equiv \begin{cases}
\frac{1}{2}+\sqrt{f_i(1-f_i)} & f_i\geq \frac{1}{2} \\
1 & f_i <\frac{1}{2}\end{cases}$, whenever the right hand side of Eq.~\eqref{app eq: result 1 bound dfadfa} is larger or equal to $1/2$.}

\begin{proof}
To start we define a particular realization of of $\varrho_{AB}, \widetilde{\t{QM}}$  $\Lambda_B^{(i)}$ and $\Lambda_B^{(o)}$. 

For the state we take $\varrho_{AB}= \Phi_\lambda$ without fixing the value $\lambda$ yet. Note however that all values $\frac{1}{2}\leq \lambda\leq \lambda_i$ can saturate
\be
\bra{\Phi^+} (\t{id}\otimes \Lambda_B^{(i)})[\Phi_\lambda]\ket{\Phi^+}  = F_i.
\ee 
To acheive equality one may take a qubit unitary map $\Lambda_B^{(i)}[\rho] = U^{(i)} \rho U^{(i)\dag}$ with 
\be \frac{1}{2}|\sqrt{\lambda}\bra{0}U^{(i)}\ket{0} +\sqrt{1-\lambda}\bra{1}U^{(i)}\ket{1}|^2 = f_i.
\ee
This is possible because the expression at the left hand side may take all values between 0 and $\frac{(\sqrt{\lambda}+\sqrt{1-\lambda})^2}{2}$, and furthermore $\frac{(\sqrt{\lambda}+\sqrt{1-\lambda})^2}{2}$ is a decreasing function of $\lambda\geq 1/2$ and  
\be
\frac{(\sqrt{\lambda_i}+\sqrt{1-\lambda_i})^2}{2} =\begin{cases}
f_i & f_i \geq 1/2 \\
1/2 & f_i < 1/2 
\end{cases} \qquad \geq f_i.
\ee

Next, for the memory channel $\widetilde{\t{QM}}$ we take the qubit map $\cE$ which minimizes the function $G(f_o,\lambda)$, and we let the extraction map to be trivial $\Lambda_B^{(o)}=\text{id}$. From the solution of the SDP we know that $\cE$ corresponds to the Choi operator 
\begin{equation} \label{app eq: sigma solution}
    \sigma = \left(\begin{array}{cccc}
    \frac{1}{2} & 0 & 0 & \frac{\sqrt{2f_o}-\sqrt{\lambda}}{2\sqrt{1-\lambda}}\\
    0 & 0 & 0 & 0\\
    0 & 0 & \frac12-\frac{(\sqrt{2f_o}-\sqrt{\lambda})^2}{2(1-\lambda)} & 0\\
    \frac{\sqrt{2f_o}-\sqrt{\lambda}}{2\sqrt{1-\lambda}} & 0 & 0 & \frac{(\sqrt{2f_o}-\sqrt{\lambda})^2}{2(1-\lambda)}\\
    \end{array}\right),
\end{equation}
which is a valid Choi state for $\lambda\leq \frac{1}{2 f_o}$. We will this further restrict the range of considered states to $\frac{1}{2}\leq \lambda\leq \min(\frac{1}{2 f_o},\lambda_i)$. The state $\sigma$ is rank-2 and can be decomposed as 
\be\begin{split}
\sigma &= \left(\frac12-\frac{(\sqrt{2f_o}-\sqrt{\lambda})^2}{2(1-\lambda)}\right) \ketbra{10} \\
    & +\left(\frac12+\frac{(\sqrt{2f_o}-\sqrt{\lambda})^2}{2(1-\lambda)}\right) \Psi\\
    \ket{\Psi}&= \frac{ \sqrt{1-\lambda}\ket{00} +(\sqrt{2 f_o}-\sqrt p)\ket{11}}{\sqrt{1+2 f_o - 2\sqrt{2f_o \lambda}}}.
\end{split}
\ee
This implies that the channel $\cE$ saturating the bound is an ``amplitude damping'' channel represented by a set of two Kraus operators
\be\begin{split}
\cE[\rho] &= K_0  \rho K_0^\dag +  K_1 \rho K_1^\dag  \\
K_0 &= \sqrt{\xi} \ketbra{0}{1} \\
 K_1 & = \sqrt{\id - K_0^\dag K_0} = 
 \left(\begin{array}{cc}
 1 & 0\\
 0 & \sqrt{1-\xi}
 \end{array}\right)\\
 \xi &\equiv 1-\frac{(\sqrt{2f_o}-\sqrt{\lambda})^2}{(1-\lambda)}.
\end{split}
\ee
One can easily verify that this choice of the state $\varrho_{AB}$, the memory 
 $\widetilde{\t{QM}} =\cE$ and extraction $\Lambda_B^{(o)}=\text{id}$ satisfy
\be
\bra{\Phi^+}(\t{id }\otimes \cE) [\Phi_\lambda] \ket{\Phi^+} = f_o,
\ee
as it is the solution of the SDP.

In summary, we have now shown that for all the realisations of the memory channel $\widetilde{\t{QM}} =\cE$ with $\frac{1}{2}\leq \lambda\leq \min(\frac{1}{2 f_o},\lambda_i)$ the relations $F_i=f_i$ and $F_o=f_o$ can be satisfied. To derive an upper bound on $\max_{\Lambda,V}\cF_{\t{id}}^{(\Lambda,V)}(\widetilde{\t{QM}})$ which holds fo all the set of channels $\widetilde{\t{QM}}$ compatible with  $F_i=f_i$ and $F_o=f_o$ it is thus sufficient to upper bound the Choi fidelity of the channel $\cE$  for all possible qubit CPTP maps $\Lambda,V$. That is upper-bounding the following function
\be
\Theta(f_o,\lambda) =\max_{\Lambda,V} \bra{ \Phi^+} (\text{id}\otimes \Lambda\circ \cE \circ V)[\Phi^+]\ket{\Phi^+}.
\ee

Let us rewrite this expression in a more convenient form. We have
\be\begin{split}
&\bra{ \Phi^+} (\text{id}\otimes \Lambda\circ \cE \circ V)[\Phi^+]\ket{\Phi^+} \\
&=\tr (\t{id}\otimes \Lambda^*)[ \Phi^+]  (\text{id}\otimes \cE \circ V)[\Phi^+] \\
 & =    \tr ( \Lambda^{T*}\otimes \t{id})[ \Phi^+]  (\text{id}\otimes \cE \circ V)[\Phi^+]
\\
& =    \tr  \Phi^+  (\Lambda^T\otimes \cE \circ V)[\Phi^+] \\
& = \tr ( \t{id}\otimes \cE^*)[ \Phi^+]  (\Lambda^T \otimes  V)[\Phi^+],
\end{split}
 \ee
 where $ \Lambda^T$ is the transposed map of $\Lambda$ (it has all Kraus opearators transposed). Let us now define and diagonalize the operator appearing in the last line
\begin{align}
\omega &\equiv (\t{id}\otimes \cE^*)[\Phi^+]\\
& = \frac{2-\xi}{2} \ketbra{\phi} + \frac{\xi}{2} \ketbra{01}\\
\ket{\phi} &= \frac{1}{\sqrt{2-\xi}} \ket{00}+ \sqrt{1-\xi}\ket{11}
\end{align}
Since, $\Lambda^T$ and $\Lambda$ run over the same set we find that 
\be
\Theta(F_o,p) =\max_{\Lambda,V} \tr \omega \, (\Lambda\otimes V)[\Phi^+].
\ee

It will be convenient to expand the states of interest = in the Pauli basis
\be\begin{split}
\Phi^+ & = \frac{1}{4}\left(\sigma_0\otimes\sigma_0 + \sigma_x\otimes\sigma_x-\sigma_y\otimes\sigma_y + \sigma_z\otimes\sigma_z \right)\\
\omega &= \sum_{i,j=0}^3 \Omega_{ij} \, \sigma_i\otimes \sigma_j \\
\Omega&= \left( \begin{array}{cccc}
1&&&\\
&{1-\xi}&&\\
&&-\sqrt{1-\xi}&\\
\xi&&&1-\xi
\end{array}\right) 
\end{split}
\ee
Recall that any single qubit CPTP map admits a representation as an affine transformation of the Block sphere
\be\begin{split}
\Lambda[\id] &= \id + \bm \sigma^\dag \bm a(\Lambda)   \\
\Lambda[ \sigma^\dag \bm v] &=\bm \sigma^\dag M_\Lambda \bm v.
\end{split}
\ee
given by a real $3$-vector $\bm a(\Lambda)$ and a $3\times 3$ real matrix $M_\Lambda$. This representation allow us to express the the state $\Phi^+$ transformed by the maps $\Lambda$ and $V$ in the Pauli basis as
\be\begin{split}
 (\Lambda \otimes  V)[\Phi^+] &=\frac{1}{4}\big(\sigma_0 + \bm \sigma^\dag \bm a(\Lambda)\big)\otimes \big(\sigma_0 + \bm \sigma^\dag \bm a(V)\big) \\ &+    \frac{1}{4} \sum_{i,j=1}^3 (M_\Lambda 
\underbrace{\left(\begin{array}{ccc}
1 & &\\
&-1& \\
&&1
\end{array}\right)}_{\equiv \mathds{J}}
M^T_V )_{ij} \sigma_i \otimes \sigma_j,
\end{split}
\ee
and expres the goal function as
\be\begin{split}
&  \tr \left(\omega (\Lambda \otimes  V)[\Phi^+]\right) = \frac{1}{4}\Biggl[1+ \xi a_3(\Lambda) + \\&\tr M_\Lambda 
\mathds{J}
M^T_V 
\mathds{J} \underbrace{\left(\begin{array}{ccc}
\sqrt{1-\xi} & &\\
&\sqrt{1-\xi}& \\
&&1-\xi
\end{array}\right)}_{\equiv \Xi }\Biggr]. 
\end{split}\ee

Now we can upper bound the last term with the help of von Neumann's trace inequality combined with H\"older's inequality for Schatten norms  
\be
\tr A B \leq \|A B\|_1 \leq \|A\|_1 \|B\|_\infty
\ee
with $\|A\|_1= \tr \sqrt{ A^\dag A}$ (the sum of singular values of $A$) and $\| B\|_\infty $ the operator norm (maximal singular value of $B$) to obtain
\be\begin{split}
\tr \mathds{J} \Xi M_\Lambda 
\mathds{J}
M^T_V 
 &\leq 
\|\mathds{J}\Xi M_\Lambda 
\|_1 \, \|\mathds{J} M_V^T \|_\infty
\\
&\leq  \|\mathds{J}\Xi M_\Lambda 
\|_1.
\end{split}
\ee
Here we used $\|\mathds{J} M_V^T\|_\infty=\|M_V^T\|_\infty \leq 1$ (a  CPTP map does not take vectors outside of the Bloch sphere). With the help of the singular value decomposition of the matrix $M_V = R D R'$ with nonegative diagonal $D$ and orthogonal $R,R'$ we can further simplify the last expression to 
\be\begin{split}
\|\mathds{J}\Xi M_\Lambda\|_1 &= \|\mathds{J}\Xi R D R'\|_1 \\
& =
\|\Xi R D \|_1
\end{split}.
\ee

Now we will use the fact that our figure of merite is linear in both maps $\Lambda$ and $V$. Hence it's maximum is attained for extremal maps, i.e.~those which can not be decomposed as convex mixture of other maps. For a single qubit extremal CPTP maps have at most two Kraus operators~\cite{verstraete2002quantum}, and their affine representation is given by~\cite{valcarce2020minimum} 
\be\begin{split}
\bm a &= R 
\left(\begin{array}{ccc} 0 &0 &\cos^2 a-
\cos^2 b\end{array}\right) = (\cos^2 a-
\cos^2 b) R 
\bm e_z\\
    M &= R D(a,b) R'\\
    D(a,b) &=
\left(\begin{array}{ccc} 
\cos (a-b)& &\\
& \cos (a+b) 
& \\
& & \cos^2 a +\cos^2 b -1
\end{array}\right)
\end{split}
\ee
with scalar parameters $a,b \in[0,\pi/2]$ and arbitrary rotation matrices $R,R'$. Note that these expression are invariant under the exchange of $a$ and $b$ together with the following modification of the rotation matrices
\begin{align}\nonumber
R \to R \left(\begin{array}{ccc} 
1& &\\
& -1 & \\
& & -1
\end{array}\right) \quad
R' \to  \left(\begin{array}{ccc} 
1& &\\
& -1 & \\
& & -1
\end{array}\right) R' .
\end{align}
Hence we can also take $a\leq b$ and $\cos^2 a \geq \cos^2 b$ without loss of generality.

Plugging these expression in the bound derived above we find
\be\begin{split}
&  \tr \left(\omega (\Lambda \otimes  V)[\Phi^+]\right) \\
&\leq \frac{1}{4}\left(1+ \xi \,  (\cos^2 a-
\cos^2 b) \bm e_z^T R \bm e_z
+\left\|\Xi
R 
D(a,b)
\right\|_1 \right).
\end{split}\ee
 One recalls that here
\be
\Xi = \left(\begin{array}{ccc}
\sqrt{1-\xi} & &\\
&\sqrt{1-\xi}& \\
&&1-\xi
\end{array}\right)
\ee
with $\sqrt{1-\xi}\geq 1-\xi$ and $R$ is an arbitrary rotation matrix. 

Let us now consider two cases separately, starting with $a+b\leq \pi/2$.  
To bound the 'vector' terms we can simply use $\bm e_z^T R \bm e_z\leq 1$.
While for the 'matrix'  term we use
\be
\|A B\|_1 = \sum_i \sigma_i(AB) \leq \sum_i \sigma_i(A) \sigma_i(B),
\ee
where $\{\sigma_1(A),\sigma_2(A),...\}$ list the singular values of operators in decreasing order $\sigma_i(A)\geq \sigma_{i+1}(A)$~\cite{marshall1979inequalities}, implying
\be
\left\|\Xi
R 
D(a,b)
\right\|_1 \leq \sum_{i=1}^3 \sigma_i(\Xi) \sigma_i(D(a,b))
\ee

For $a+b\leq \pi/2$ (or $\cos^2 a +\cos^2 b\leq 1$) one easily verifies that the singular values of $D(a,b)$ satisfy
$\cos(a-b)\geq \cos(a+b)\geq \cos^2 a + \cos^2 b -1 \geq 0$ and 
\be\begin{split}
& \left\|\Xi R 
D(a,b)
\right\|_1 \\ 
&\leq 2 \cos^2 a \cos^2 b\sqrt{1-\xi} + (\cos^2 a + \cos^2 b -1)(1- \xi).
\end{split}\ee
Adding the other terms we find
\be \label{app: ub 2}\begin{split} &\tr \left(\omega (\Lambda \otimes  V)[\Phi^+]\right) 
\\&\leq 
\frac{1}{4}\Big(1+  (\cos^2 a - \cos^2 b)\xi  \\
&+2 \cos^2 a \cos^2 b\sqrt{1-\xi} + (\cos^2 a + \cos^2 b -1)(1- \xi) \Big)
\\
&\leq \frac{1}{4}\big(1+\xi + \cos^2 b(1+ 2 \sqrt{1-\xi}-2\xi) \big),
\end{split}
\ee
where $\cos^2 a$ comes with a positive factor so we set $\cos^2 a = 1$ to obtain an upper bound. The maximum of the last expression depends of the sign of the term  $(1-2\xi +2\sqrt{1-\xi})$ in front of the $\cos^2 b$ term, and is attained by setting $\cos^2 b=0$ or $1$ (in agreement with $a+b\leq \pi/2$). Concretely, we obtin
\be
 \tr \omega (\Lambda \otimes  V)[\Phi^+]\leq \begin{cases}
\frac{1}{4}\big(2-\xi + 2 \sqrt{1-\xi}\big) & \xi\leq \frac{\sqrt 3}{2}\\
\frac{1}{4}\big(1+\xi\big) & \xi> \frac{\sqrt 3}{2}\\
\end{cases}.
\ee

For the case $a+b > \pi/2$  (or $ \sin^2 b + \sin^2 a \geq 1$) the singular values of $D(a,b)$ are
$\cos(a-b)\geq -\cos(a+b)\geq 1-\cos^2 a - \cos^2 b \geq 0$, yielding for the matrix term
\be\begin{split}
& \left\|\Xi R 
D(a,b)
\right\|_1 \\ 
&\leq 2 \sin^2 a \sin^2 b\sqrt{1-\xi} + (\sin^2 a + \sin^2 b -1)(1- \xi).
\end{split}
\ee
In addition we can also rewrite $\cos^2 a - \cos^2 b = \sin^2 b -\sin^2 a$ and express the overall bound as 
\be\label{eq app: upper bound} \begin{split} &\tr \left(\omega (\Lambda \otimes  V)[\Phi^+]\right) 
\\&\leq 
\frac{1}{4}\Big(1+  (\sin^2 b - \sin^2 a)\xi  \\
&+2 \sin^2 a \sin^2 b\sqrt{1-\xi} + (\sin^2 a + \sin^2 b -1)(1- \xi) \Big).
\end{split}
\ee
We recover the expression in Eq.~\eqref{app: ub 2} upon relabeling $\sin^2 b \to \cos^2 a$ and  $\sin^2 a \to \cos^2 b$. Hence, the upper bound derived for $a+b\leq \pi/2$ also holds here.

Finally, let us have a closer look at the upper bound ~\eqref{eq app: upper bound} that we found on the Choi fidelity of the map $ \cE$ that we have found
\[
\Theta(F_o,p) \leq
\begin{cases}
\frac{1}{4}\big(2-\xi + 2 \sqrt{1-\xi}\big) & \xi\leq \frac{\sqrt 3}{2}\\
\frac{1}{4}\big(1+\xi\big) & \xi> \frac{\sqrt 3}{2}\\
\end{cases}
\]
Remarkably, in the first branch ($\xi\leq \frac{\sqrt 3}{2}$) the expression at the right hand side becomes
\[
\frac{2-\xi + 2 \sqrt{1-\xi}}{4}=\left( \frac{\sqrt{2f_o}-\sqrt{\lambda}+\sqrt{1+\lambda}}{2\sqrt{1+\lambda}}\right), 
\]
which is tight with the lower-bound $G(f_o,\lambda)$. Furthermore, the other branch ($\xi> \frac{\sqrt 3}{2}$) corresponds to a regime where $G(f_o,\lambda),\Theta(f_o,\lambda) \leq \frac{1}{2}$.  Hence, we have $\Theta(f_o,\lambda) = G(f_o,\lambda)$ whenever $G(f_o,\lambda)\leq \frac{1}{2}$. Now since we are looking for an upper bound we will chose the value $\lambda$ in its range of validity $\frac{1}{2}\leq \lambda\leq \min(\frac{1}{2 f_o},\lambda_i)$ to minimize $\Theta(f_o,\lambda)$. But this is exactly how the function $F'(f_o,f_i)$ was constructed from $G(f_o,\lambda)$. We this conclude that
\be
\min_{\text{valid }\lambda } \max_{\Lambda,V}\cF_{\t{id}}^{(\Lambda,V)}(\widetilde{\t{QM}}) \leq \min_{\text{valid }\lambda } \Theta(f_o,\lambda) = F'(f_o,f_i)
\ee
whenever $F'(f_o,f_i)\geq 1/2$, which concludes the proof.
\end{proof}

\subsection{Scenario 3} \label{sec:scenario3}

\noindent \textbf{Result 3.} 
\textit{Consider a state $\varrho_{AB}\in B(\mathds{C}_A^2\otimes \cH_B^{(i)})$ and CPTN maps $\widetilde{\t{QM}},$  $\bar \Lambda_B^{(i)}$ and $\bar \Lambda_B^{(o)}$ such that}
\begin{align}\label{eq: Fi app}
\bra{\Phi^+} \frac{(\t{id}\otimes \bar{\Lambda}_B^{(i)})[\varrho_{AB}]}{\tr (\t{id}\otimes \bar{\Lambda}_B^{(i)})[\varrho_{AB}]}\ket{\Phi^+}&\geq f_i
\\  
\label{eq: Pi app}
\tr (\t{id}\otimes \bar{\Lambda}_B^{(i)})[\varrho_{AB}] &\geq p_i \\ \label{eq: Fo app}
\bra{\Phi^+}\frac{(\t{id}\otimes \bar{\Lambda}_B^{(o)}\circ \widetilde{\t{QM}})[\varrho_{AB}]}{\tr (\t{id}\otimes \bar{\Lambda}_B^{(o)}\circ \widetilde{\t{QM}})[\varrho_{AB}]} \ket{\Phi^+} =F_o &\geq f_o \\
\label{eq: Po app}
 \tr ( \t{id}\otimes \bar{\Lambda}_B^{(o)}\circ \widetilde{\t{QM}})[\varrho_{AB}] =P_o &\geq p_o.
\end{align}
\textit{ There exist a CPTN map $\bar{\Lambda}$ and a CPTP map $V$ such that 
\be\label{eq: res 2}
\begin{split}
    \cF_{\t{id}}^{(\bar{\Lambda},V)}(\widetilde{\t{QM}}) &\geq  \text{f}_{\min}(f_o,f_i,p_o,p_i)  \\
     P_{\checkmark}^{(\bar \Lambda,V)}(\widetilde{\t{QM}}) &\geq  \text{p}_{\min}(f_o,f_i,p_o,p_i),
\end{split}
\ee
where the functions $\text{f}_{\min}(f_o,f_i,p_o,p_i)$ and $\text{p}_{\min}(f_o,f_i,p_o,p_i)$ are defined as the results of simple numerical minimization in Eqs.~\eqref{eq app: fmin} and \eqref{eq app: pmin} respectively.}

\begin{proof}
By lemma~\ref{lemma2} Eqs.~(\ref{eq: Fi app},\ref{eq: Pi app}) gurantee that 
\be \label{eq: 92}
\varrho_{AB} = u \otimes V [\Phi_\lambda]
\ee
for a qubit unitary $u$, a CPTP map $V$, and the state $\ket{\Phi_\lambda} =\sqrt{\lambda} \ket{00} + \sqrt{1-\lambda} \ket{11}$ with
\be
\frac{1}{2}\leq \lambda \leq \lambda_i = \begin{cases} 
1-\left(\frac{1}{2} -\sqrt{f_i(1-f_i)}\right) p_i & f_i\geq \frac{1}{2} \\
1 & f_i < \frac{1}{2}
\end{cases}.
\ee

For the moment let us view $\lambda$ as a fixed value. Eq.~\eqref{eq: 92} allows us to rewrite the bounds~(\ref{eq: Fo app},\ref{eq: Po app}) as
\begin{align}
    \bra{\Phi^+}(\t{id}\otimes u^T \circ  \bar{\Lambda}_{B}^{(o)} \circ \widetilde{\t{QM}} \circ V) [\Phi_\lambda] \ket{\Phi^+} &= P_o F_o \\
    \tr (\t{id}\otimes u^T \circ  \bar{\Lambda}_{B}^{(o)} \circ \widetilde{\t{QM}} \circ V)[\Phi_\lambda]  = P_o.
\end{align}
For short we define the CPTN map $\bar{\cE} \equiv u^T \circ  \bar{\Lambda}_{B}^{(o)} \circ \widetilde{\t{QM}} \circ V$ and rewrite the above bounds as 
\begin{align} \label{eq: 96}
   \tr (\t{id}\otimes \bar \cE) [\Phi_\lambda] \, \Phi^+ &= P_o F_o \\
   \label{eq: 97}
    \tr (\t{id}\otimes \bar \cE)[\Phi_\lambda]  &= P_o.
\end{align}
This is close to what we want. There is however an important difference -- for the quantities $\cF_{\t{id}}^{(\bar{\Lambda},V)}(\widetilde{\t{QM}})$ and   $P_{\checkmark}^{(\bar \Lambda,V)}(\widetilde{\t{QM}})$ the map $\bar \cE$ is to act on $\Phi^+$ and not $\Phi_\lambda$.

Nevertheless, the two states are related by the following qubit CPTN map $K_A$ acting on Alice's side
\be\label{eq: 98}
\Phi_\lambda = \frac{(K_A\otimes\t{id})[\Phi^+]}{\tr (K_A\otimes\t{id})[\Phi^+]}
\ee
where the $K_A\otimes \t{id} : \rho \mapsto T_A \rho T_A$ with
\be
T_A=\left(\begin{array}{cc}
1 & \\
& \sqrt{t} 
\end{array}\right)\otimes \mathds{1}
\ee
for $t = \frac{1-\lambda}{\lambda}$ and
\be\label{eq: 100}
\tr (K_A\otimes\t{id})[\Phi^+] = \frac{1+t}{2} .
\ee

Let us now consider the following operator 
\be
\sigma = (\t{id}\otimes \bar \cE)[\Phi^+].
\ee
It is the Choi state of the map $\bar \cE$, and  satisfies 
\be \label{eq: choi proba}\begin{split}
&\sigma_A = \tr_B \sigma \preceq \frac{1}{2} \id \\
& \sigma \succeq 0
\end{split}\ee
by construction. Furthermore by Eqs.~(\ref{eq: 98},\ref{eq: 100}) we have
\be\begin{split}
T_A \sigma T_A &=(K_A\otimes \t{id})[\sigma]\\
&=(K_A\otimes \bar \cE)[\Phi^+] 
\\
&=(\t{id}\otimes \bar \cE)\circ (K_A\otimes \t{id})[\Phi^+]\\
&= \frac{1+t}{2} (\t{id}\otimes \bar \cE)[\Phi_\lambda],
\end{split}
\ee
which implies 
\be\label{eq: bounds boudns}\begin{split}
\tr T_A \sigma T_A &=\frac{1+t}{2} P_o \\
\bra{\Phi^+} T_A \sigma T_A \ket{\Phi^+} &=  \frac{1+t}{2} F_o P_o.
\end{split}
\ee
by virtue of Eqs.~(\ref{eq: 96},\ref{eq: 97}). One  sees that the state $\frac{(K_A\otimes \bar \cE)[\Phi^+]}{\tr (K_A\otimes \bar \cE)[\Phi^+]}$ has singlet fidelity $F_o$. However, it can not be used to certify the memory because it requires a filtering to be applied on Alice's system.

Instead let us now consider the state

\be
\varrho=(\t{id}\otimes K_B)[\sigma] = T_B \sigma T_B = (\t{id}\otimes K_B\circ\bar \cE)[\Phi^+]
\ee
with
\be
T_B = \id \otimes \left(\begin{array}{cc}
1 & \\
& \sqrt{t} 
\end{array}\right).
\ee
This is the Choi state of the CPTN map $K_B \circ \bar \cE = K_B \circ u^T \circ  \bar{\Lambda}_{B}^{(o)} \circ \widetilde{\t{QM}} \circ V$. For the choice of the probabilistic extraction map $\bar \Lambda \equiv K_B \circ u^T \circ  \bar{\Lambda}_{B}^{(o)}$ we obtain
\be
\varrho = (\t{id}\otimes \bar \Lambda \circ \widetilde{\t{QM}}  \circ V)[\Phi^+],
\ee
so that the following quantities
\be\begin{split}
A &= \tr\Phi^+\,  \varrho  =\bra{\Phi^+} T_B\sigma T_B \ket{\Phi^+}\\
B &=\tr \varrho = \tr  T_B \sigma T_B, 
\end{split}
\ee
imply a certification of the memory of the form
\begin{align}
\cF_{\t{id}}^{(\bar{\Lambda},V)}(\widetilde{\t{QM}}) &\geq \frac{A}{B}  \\
     P_{\checkmark}^{(\bar \Lambda,V)}(\widetilde{\t{QM}}) &\geq B.
\end{align}
Now noting that $T_B \ket{\Phi^+} = T_A \ket{\Phi^+}$ we get that
\be
A =\bra{\Phi^+} T_A \sigma T_A \ket{\Phi^+}= \frac{1+t}{2} F_o P_o \geq \frac{1+t}{2} f_o p_o
\ee
from Eq.~\eqref{eq: bounds boudns}.

It remains to bound $B$. Using the inequalities~(\ref{eq: bounds boudns},\ref{eq: choi proba}), both the minimal and the maximal possible value of $B$ can be found as solutions of an SDP.  In particular, the maximal value that $B$ can take can be expressed as the following SDP: 
\be\begin{split}\label{eq:SDPmax}
B_\t{max}(f_o,p_o,t)\equiv \max_\sigma \, &\tr T_ B\sigma  T_B \\
\t{s.t.}\, &\tr \sigma T_A \Phi^+ T_A \geq \frac{1+t}{2} f_o p_o\\
& \tr T_A \sigma  T_A \geq \frac{1+t}{2} p_o  \\
&\sigma_A = \tr_B \sigma \preceq \frac{1}{2} \id \\
& \sigma \succeq 0.
\end{split}
\ee
Analogously, the minimal value that $B$ can take can be expressed in an SDP form as: 
\be\begin{split}\label{eq:SDPmin}
B_\t{min}(f_o,p_o,t)\equiv \min_\sigma \, &\tr T_ B\sigma  T_B \\
\t{s.t.}\, &\tr \sigma T_A \Phi^+ T_A \geq \frac{1+t}{2} f_o p_o\\
& \tr T_A \sigma  T_A \geq \frac{1+t}{2} p_o  \\
&\sigma_A = \tr_B \sigma \preceq \frac{1}{2} \id \\
& \sigma \succeq 0.
\end{split}
\ee

This allows us to conclude that the values $\cF_{\t{id}}^{(\bar{\Lambda},V)}(\widetilde{\t{QM}})$ and $ P_{\checkmark}^{(\bar \Lambda,V)}(\widetilde{\t{QM}})$ satisfy
\begin{align}
\label{eq app: 122} \cF_{\t{id}}^{(\bar{\Lambda},V)}(\widetilde{\t{QM}}) &\geq \frac{1+t}{2}\frac{f_o p_o}{B}  \\
\label{eq app: 123} 
     P_{\checkmark}^{(\bar \Lambda,V)}(\widetilde{\t{QM}}) &\geq B.
\end{align}
for some $B\in[B_\t{min}(f_o,p_o,t), B_\t{max}(f_o,p_o,t)]$, and some $t\in \left[\frac{1-\lambda_i}{\lambda_i},1\right]$.

Simpler but relaxed bounds on the fidelity and success probability of the memory can be obtained by minimizing the expressions at the right hand side of Eqs.~(\ref{eq app: 122},\ref{eq app: 123}) independently. We then obtain
\be
\begin{split}
\cF_{\t{id}}^{(\bar{\Lambda},V)}(\widetilde{\t{QM}}) &\geq \text{f}_{\min}(f_o,f_i,p_o,p_i) \\
P_{\checkmark}^{(\bar \Lambda,V)}(\widetilde{\t{QM}}) &\geq \text{p}_{\min}(f_o,f_i,p_o,p_i) 
\end{split}
\ee
for the functions
\begin{align}\label{eq app: fmin}
 \text{f}_{\min}(f_o,f_i,p_o,p_i) & \equiv \min_{t \in \left[\frac{1-\lambda_i}{\lambda_i},1\right]} \frac{(1+t)f_o p_o}{ 2 B_\t{max}(f_o,p_o,t)} \\
 \label{eq app: pmin}
    \text{p}_{\min}(f_o,f_i,p_o,p_i) & \equiv \min_{t \in \left[\frac{1-\lambda_i}{\lambda_i},1\right]}  B_\t{min}(f_o,p_o,t),
\end{align}
as stated in the result.
\end{proof}

\end{document}